\documentclass{article}
\usepackage[utf8]{inputenc}
\usepackage{geometry}
\usepackage{color}
\usepackage{graphicx}
\usepackage{tabularx, environ}
\usepackage{tikz}

\usepackage{verbatim}
\usepackage{mathrsfs}
 \usepackage{rotating}
\usepackage{enumerate}
\usepackage{cite}
\usepackage{algorithm}
\usepackage{amsmath, amsthm, amssymb}
\newtheorem{theorem}{Theorem}[section]
\newtheorem{prop}{Proposition}[section]

\newtheorem{remark}[theorem]{Remark}

\newtheorem{definition}[theorem]{Definition}
\usepackage[noend]{algpseudocode}
\usepackage[sort&compress]{natbib}
\usepackage{soul}
\usepackage{caption}
\usepackage{subcaption}
\usepackage{lscape}
\usepackage{multirow}
\usepackage{fixmath}
\bibpunct{(}{)}{;}{a}{}{,}

\makeatletter
\newcommand{\problemtitle}[1]{\gdef\@problemtitle{#1}}
\newcommand{\probleminput}[1]{\gdef\@probleminput{#1}}
\newcommand{\problemquestion}[1]{\gdef\@problemquestion{#1}}

\NewEnviron{problem}{
  \problemtitle{}\probleminput{}\problemquestion{}
  \BODY
  \par\addvspace{.5\baselineskip}
  \noindent
  \begin{tabularx}{\textwidth}{@{\hspace{\parindent}} l X c}
    \multicolumn{2}{@{\hspace{\parindent}}l}{\@problemtitle} \\
    \textbf{Input:} & \@probleminput \\
    \textbf{Question:} & \@problemquestion
  \end{tabularx}
  \par\addvspace{.5\baselineskip}
}
\newcommand\notsotiny{\@setfontsize\notsotiny{6.3}{6.3}}
\makeatother

\title{On path ranking in time-dependent graphs}
\author{Tommaso Adamo, Gianpaolo Ghiani, Emanuela Guerriero,\\
Dipartimento di Ingegneria per l'Innovazione, \\Universit\`{a} del Salento, Lecce, Italia}
\date{}

\begin{document}

\maketitle
\begin{abstract}
In this paper we study a property of time-dependent graphs, dubbed "path ranking invariance''. Broadly speaking, a time-dependent graph is "path ranking invariant'' if the ordering of its paths (w.r.t. travel time) is independent of the start time. In this paper we show that, if a graph is path ranking invariant, the solution of a large class of time-dependent vehicle routing problems can be obtained by solving suitably defined (and simpler) time-independent routing problems. We also show how this property can be checked by solving a linear program. If the check fails, the solution of the linear program can be used to determine a tight lower bound. In order to assess the value of these insights, the lower bounds have been embedded into an enumerative scheme. Computational results on the time-dependent versions of the \textit{Travelling Salesman Problem} and the \textit{Rural Postman Problem} show that the new findings allow to outperform state-of-the-art algorithms.
 
\end{abstract}

\section{Introduction} \label{intro}
Vehicle routing is concerned with the design of "least cost" routes for a fleet of vehicles, possibly subject to side constraints, such as vehicle capacity or delivery time windows. 
According to \cite{toth2014vehicle}, the vast majority of the literature is based on the assumption that the data used to formulate the problems do not depend on time. Only in recent years there has been a flourishing of scholarly work in a time-dependent setting (\cite{gendreau2015time}). \\
\indent The main goal of this paper is to study a fundamental property of time-dependent graphs that we call \textit{path ranking invariance}. As detailed in the following, a time-dependent graph is path ranking invariant if the ordering of its paths w.r.t. travel duration is not dependent on the start travel time. We demonstrate that this property can be exploited to solve a large class of time-dependent routing problems including the \textit{Time-Dependent Travelling Salesman Problem} (TDTSP) and the \textit{Time-Dependent Rural Postman Problem} (TDRPP).
We prove that if a graph is path ranking invariant, then the ordering of the solutions of these problems, w.r.t. travel duration, is the same as a suitably-defined time-independent counterpart. In order to determine sufficient conditions for path ranking invariance, we introduce a \textit{decision problem}, named the \textit {Constant Traversal Cost Problem} (CTCP). A \textit{decision problem} is a problem with a yes-or-no answer (\cite{arora2009computational}). We prove that if a time-dependent graph is a yes-instance of the CTCP, then it is path ranking invariant. Then, the decidability of the CTCP is demonstrated by devising a certificate-checking algorithm based on the solution of a linear program. Finally, we show that, if the CTTP feasibility check fails, our results can be used to determine an auxiliary path ranking invariant graph, where the travel time functions are lower approximations of the original ones. Such less congested graph can be used to determine lower bounds. We evaluate the benefits of this new approach on both the TDTSP and the TDRPP.

\indent The paper is organized as follows. Section \ref{rev} summaries the literature. In Section \ref{sec_1}, we formally define the path ranking invariance property and discuss its relationship with optimality conditions of the \textit{Time-Dependent General Routing Problem} (TDGRP) which includes the TDTSP and TDRPP as special cases.  In Section \ref{sec_2}, we introduce a parameterized family of travel cost functions and define a set of sufficient conditions for path ranking invariance. In Section \ref{sec_3}, we define the  Constant Traversal Time Problem. In particular, we prove the decidability of the CTCP by devising an algorithm that correctly decides if a time-dependent graph is a yes-instance. In Section \ref{sec_4}, we define a lower bounding procedure. In Section \ref{sec_5},  we discuss the benefits obtained when this new approach is embedded into state-of-the-art algorithms for the the TDTSP and TDRPP.  Finally, some conclusions follow in Section \ref{sec_6}.

\section{State of the art}\label{rev}
The literature on time-dependent routing problems is quite scattered and disorganized. In this section, we present a brief review of the \textit{Time-Dependent Travelling Salesman Problem} (TDTSP) and the \textit{Time-Dependent Rural Postman Problem} (TDRPP) which are used in this paper to test the computational potential of the path ranking invariance property. For a complete survey on time-dependent routing, see \cite{gendreau2015time}.

\subsection{Time-Dependent Travelling Salesman Problem}
 \cite{malandraki1992} were the first to address the TDTSP and proposed a \textit{Mixed Integer Programming} (MIP) formulation. Then \cite{Malandraki199645} 
devised an approximate dynamic programming algorithm while \cite{Li2005} developed two heuristics. \cite{schneider2002} proposed a simulated annealing heuristic and \cite{Harwood:2013aa} presented some metaheuristics.
\cite{Cordeau2014} derived some properties of the TDTSP as well as lower and upper bounding procedures. They also represented the TDTSP as MIP model for which they developed some families of valid inequalities. These
inequalities were then used into a branch-and-cut algorithm that solved instances with up to 40 vertices. 
\cite{arigliano2018branch} exploited some properties of the problem and developed a branch-and-bound
algorithm which outperformed the \cite{Cordeau2014} branch-and-cut procedure.  \cite{melgarejo2015time} presented a new global constraint that was used in a \textit{Constraint Programming} approach. This algorithm was able to solve instances with up to 30 customers. Recently, \cite{ADAMO2020104795} proposed a parameterized family of lower bounds, whose parameters are chosen by fitting the traffic data. When embedded into a branch-and-bound procedure, their  lower bounding mechanism allows to solve to optimality a larger number of instances than \cite{arigliano2018branch}.\\ 
\indent Variants of the TDTSP have been examined by \cite{Albiach2008789}, \cite{arigliano2018time}, \cite{montero2017integer} and \cite{Boland2020} (\textit{TDTSP with Time Windows}), by \cite{Helvig2003153} (\textit{Moving-Target TSP}) and by \cite{Montemanni1528314} (\textit{Robust TSP with Interval Data}). \\
Finally, it is worth noting that a scheduling problem, other than the above defined TDTSP, is also known as
Time-Dependent Travelling Salesman Problem. It amounts to sequence a set of jobs on a single machine in which the processing times depend on the position of the jobs within the schedule (\cite{picard1978}, \cite{fox1980},
\cite{gouveia1995}, \cite{vanderwiel1996},
\cite{miranda2010}, \cite{stecco2008branch}, \cite{godinho2014natural}).\\

\subsection{Time-Dependent Rural Postman Problem}
\cite{Tan2011} were the first to propose an exact algorithm for the TDRPP. They devised an integer linear programming model, based on an arc-path formulation enforced by the introduction of valid inequalities. The computational results showed that no instance was solved to optimality, with a relative gap between the best feasible solution and the lower bound equal to $3.16\%$ on average. \cite{calogiuri2019branch} provided both a lower bound and an upper bound with a worst-case guarantee. The proposed bounds were embedded into a branch-and-bound algorithm that was able to solve TDRPP instances with up to 120 arcs with a percentage of required arcs equal to $70\%$.\\
\indent When the set of required vertices is empty, but all arcs of the time-dependent graph have to be visited, the TDGRP reduces to the  \textit{Time-Dependent Chinese Postman problem} (TDCPP). An integer programming formulation for solving the TDCPP  was proposed by \cite{sun2011new}. \cite{sun2011dynamic} also proved that the TDCPP is NP-hard and they proposed a dynamic programming algorithm for solving it. A linear integer programming formulation, namely the cycle-path formulation, was presented by \cite{sun2015integer}.

\section{The path ranking invariance property} \label{sec_1}
Let $G := (V,A,\tau)$ be a directed and connected graph, where $V$ is the set of vertices, $A := \{(i, j) : i \in V, j \in V\}$ is the set of arcs. Moreover, let $\tau:A\times\mathbb{R}^+\rightarrow \mathbb{R}$ denote a function that associates to each arc $(i,j)\in A$ and starting time $t\in[0,+\infty)$  the traversal time when a vehicle leaves the vertex $i$ at time $t$. In particular\textcolor{black}{,} we suppose that it is given a planning horizon $[0,T]$ and the travel time functions are constant in the long run, that is $\tau(i,j,t):=\tau(i,j,T)$ with $t\geq T$. For the sake of notational simplicity, we use $\tau_{ij}(t)$ to designate $\tau(i,j,t)$. 
We suppose that  traversal time $\tau_{ij}(t)$ satisfy the \textit{first-in-first-out} (FIFO) property, i.e., leaving the vertex $i$ later implies arriving later at vertex $j$. In the following we denote with $\mathcal{T}_{\tau}(i,j)$, the ordered set of time instants corresponding to the breakpoints of $\tau_{ij}(t)$.\\
\indent For any given path $p_k := (i_0, i_1,\dots, i_k)$, the corresponding duration  $z_\tau(p_k , t)$ can be computed recursively as:
\begin{equation}\label{z_tau}
z(p_k,t):=z(p_{k-1},t)+\tau_{i_{k-1}i_k}(z(p_{k-1},t)),
\end{equation}
with the initialization $z(p_0,t):=0$. 
\begin{definition}[\textbf{Path dominance rule}]
Given two paths $p'$ and $p''$ of $G$ and their traversal time functions, $z(p',t)$ and $z(p'',t)$ respectively, we say that $p'$ dominates $p''$, iff:
\begin{equation}
z(p',t)\geq z(p'',t)\quad \forall t\geq 0.
\end{equation}
\end{definition}
\begin{definition}[\textbf{Path ranking invariance}]
\textit{A time-dependent graph $G$ is path ranking invariant}, if the path dominance rule holds true  for any pair of paths $p'$ and $p''$ of $G$.
\end{definition}

The importance of path ranking invariance property is related to its relationship with the optimality conditions of some classical time-dependent routing problems. Given  a time-dependent graph $G:=(V,A,\tau)$, a set of required vertices $V_R\subseteq V$ and a set of required arcs $A_R\subseteq A$, let us denote with $\mathcal{P}$  a set of paths starting from and ending to a given vertex $i_0$ of $V$ and passing through each required vertex $i\in V_{R}$ and each required arc $(i,j)\in A_R$ at least once. 
Given a starting time $t_0$,  we focus on the \textit{Time-Dependent General Routing Problem} aiming to determine the least duration path on $\mathcal{P}$, that is
 \begin{equation}\label{opt_ref}
\min\limits_{p\in\mathcal{P}}z(p,t_0).
\end{equation}
For notational convenience, we model the \textit{Time-Dependent Travelling \textcolor{black}{S}alesman Problem} as a special case of the compact formulation (\ref{opt_ref}), where it is required that $G$ is complete, $A_R:=\emptyset$ and $V_R:=V$.
Algorithms developed for the time-invariant counterpart of such routing problems are not able to consider time-varying travel times without essential structural modifications. Nevertheless, we observe that the absence of time  constraints implies that time-varying travel times have an impact  on the ranking of solutions of the routing problem (\ref{opt_ref}), but they do not pose any difficulty for feasibility check of solutions. 
In particular, one can assert that there always exists a time-invariant (dummy) cost function $d:A\rightarrow  \mathbb{R}^+$ such that a least duration path of (\ref{opt_ref}) is also a least cost path of the time-invariant instance of  
\begin{equation}\label{opt_ref1}
\min\limits_{p\in\mathcal{P}}\sum\limits_{(i,j)\in p} d(i,j),
\end{equation}
where the notation $(i,j)\in p$ means that the arc $(i,j)\in A$ is traversed by the path $p$. 
\begin{definition} \label{def2} 
A time-invariant cost function $d:A\rightarrow \mathbb{R}^+$ is valid for $G$, if the least duration path of the time-dependent  instance $(\mathcal{P},\tau)$ of (\ref{opt_ref}) is also a least cost path of the time-invariant  instance $(\mathcal{P},d)$ of (\ref{opt_ref1}), for any set $\mathcal{P}$ of paths defined on $G$.
\end{definition}
If we are given a cost function \textit{valid} for a time-dependent graph $G$, then we can solve various time-dependent routing problems defined on $G$, by exploiting algorithms developed for their time-invariant counterpart. For example we can determine the least duration \textcolor{black}{H}amiltonian circuit of $G$, by solving a simpler (yet NP-Hard) TSP where cost of arc $(i,j)\in A$ is equal to $d(i,j)$. Similarly, the least duration solution of an instance of TDGRP defined on $G$ can be determined by solving a time-invariant GRP. Nevertheless, the main issue of such approach is: \textit{how to certificate that a cost function $d:A\rightarrow \mathbb{R}^+$ is valid for a given time-dependent graph}. Since the travel time functions are constant in the long run, the answer is quite straightforward for path ranking invariant graphs. 
\begin{remark}
 If the time-dependent graph $G$ is path ranking invariant, then it is valid for $G$ any cost function $d:A\rightarrow \mathbb{R}^+$, where each value $d(i,j)$ is proportional to the traversal time of arc $(i,j)\in A$ when the vehicle leaves vertex $i$ at time instant $T$, that is:
 $$\arg\min\limits_{p\in\mathcal{P}} \sum\limits_{(i,j)\in\mathcal{P}}\tau_{ij}(T)=\arg\min\limits_{p\in\mathcal{P}}z(p).$$
\end{remark}

This explain our previous assertion about the relationship between the path ranking invariance property and the optimality conditions of the class of time-dependent routing problems (\ref{opt_ref}).
In the following sections we demonstrate that a class of path ranking invariant graphs is computable. \textcolor{black}{In computability theory, a set of symbols is computable if there exists  an algorithm that correctly decides whether a symbol  belongs to such set.} In particular, we aim to devise an algorithm that takes as input a time-dependent graph $G$  and either decides correctly that $G$ belongs to \textcolor{black}{a class of path ranking invariant graphs} or returns a constant cost function suitable to determine a lower bound on the optimal solution of  (\ref{opt_ref}).

\section{A family of travel cost functions} \label{sec_2}
In this section, we propose a parameterized family of travel cost functions and investigate its relationship with path ranking invariance property. The parameter of such a cost model is a step cost function. In particular, we are interested in time-dependent graphs, for which there exists \textit{a step function generating a constant travel cost for each arc}. In subsection \ref{sec_2_1}, we prove that \textit{the existence of such step function is a sufficient condition for stating that the graph is path ranking invariant}. In subsection \ref{sec_2_2}, we discuss some properties of the proposed travel cost model. Such properties are exploited in Section \ref{sec_3}, in order to devise an algorithm that, given as input a time-dependent graph, correctly decides if there exists or not a step function \textit{generating}  a constant travel cost for each arc.

\subsection {The model}\label{sec_2_1}
Let $b:\mathcal{T}_b\rightarrow \mathbb{R}^+$ denote a  step function such that $\mathcal{T}_{b}$ is an ordered set of time instants and 
$$b(t):=b_h\quad t\in[t_h,t_{h+1}],$$
with \textcolor{black}{$b_h>0$},  $t_h\in\mathcal{T}_{b}$ and $h=0,\dots,|\mathcal{T}_{b}|-1$. For notational convenience, we also use $\textbf{b}=[(t_0,b_0),\dots,(t_{|\mathcal{T}_{b}|-1},b_{|\mathcal{T}_{b}|-1})]$ to designate the step function $b(t)$. The step-function $\mathbf{b}$ is the input parameter of a family of travel cost functions $c_{ij}(t,\textbf{b})$ defined as follows. The set  of  breakpoints  $\mathcal{T}_b$ represents  a partition of the planning horizon in $|\mathcal{T}_{b}|$ time intervals, whilst each value $b_h$ models the cost associated to one time unit spent traveling during the $h$-th time interval, with $h=0,\dots,|\mathcal{T}_{b}|-1$. \textcolor{black}{Since each unit cost  is strictly positive, it  never pays to wait.}
Given a time instant $t\geq 0$, the value $c_{ij}(t,\textbf{b})$ represents the overall travel cost associated to  arc $(i,j)\in A$,  when the vehicle leaves the vertex $i$ at time $t\in [t_p,t_{p+1}]$ and arrives at vertex $j$ at time $(t+\tau_{ij}(t))\in[t_q,t_{q+1}]$, that is:
\begin{equation}\label{c_def}
c_{ij}(t,\textbf{b}):=(t_{p+1}- t)b_{p}+\sum\limits_{h:=p+1}^{q-1}(t_{h+1}-t_h)b_{h}+(t+\tau_{ij}(t)-t_q)b_{q},
\end{equation}
with $(t_p,b_p), (t_q,b_q)\in \mathbf{b}$ and $ p,q=0,\dots,|\mathcal{T}_{b}|-1.$ In Figure \ref{Figure_1} it is reported a numerical example consisting of a given  travel time function $\tau_{ij}(t)$ and an arbitrarily chosen step function $\textbf{b}$, where the sets of breakpoints are $\mathcal{T}_\tau(i,j):=\{0.0,4.0,5.0\}$ and $\mathcal{T}_b:=\{0.0,1.0,2.0,3.0,4.0,5.0\}$.  
As shown in  Figure \ref{Figure_1}, the selected step function $\textbf{b}$ generates a travel cost function $c_{ij}(t,\textbf{b})$ which has  a constant value equal to 3. \\
\indent Let us denote with $\mathcal{L}_C$ the set of time-dependent graphs such that there exists a step function $\mathbf{b}^*$ \textit{generating  a constant travel cost }for  each arc $(i,j)\in A$.


\textcolor{black}{
\begin{prop} \label{prop_rank}
If a time-dependent graph  $G$ belongs to  $\mathcal{L}_C$, then it is path ranking invariant.\\
\end{prop}
\begin{proof}
Since $\mathbf{b}^*$ is a step function, (\ref{c_def}) can be rewritten as follows:
$$\underline{c}_{ij}=\int_t^{t+\tau_{ij}(t)}b^*(\mu)d\mu.$$
We observe that for each path $p_k$ defined on $G$ it is associated a travel cost  computed as follows:
\begin{equation}\label{z_c_def}
\sum\limits_{(i,j)\in p_k}\underline{c}_{ij}=\int_{t}^{z_\tau(p_k,t)}b^*(\mu).\nonumber
\end{equation}
This implies that for any pair of paths $p$ and $p^{\prime}$ defined on $G$, it results that:
$$\sum\limits_{(i,j)\in p}\underline{c}_{ij}\leq \sum\limits_{(i,j)\in p^{\prime}}\underline{c}_{ij}\Leftrightarrow z_\tau(p,t)\leq z_\tau(p^{'},t),$$
for $t\geq 0$. Since paths $p$ and $p^{\prime}$ have been arbitrarily chosen, the thesis is proven.
\end{proof} }

\begin{figure}[H] 

 \includegraphics[width=150mm ]{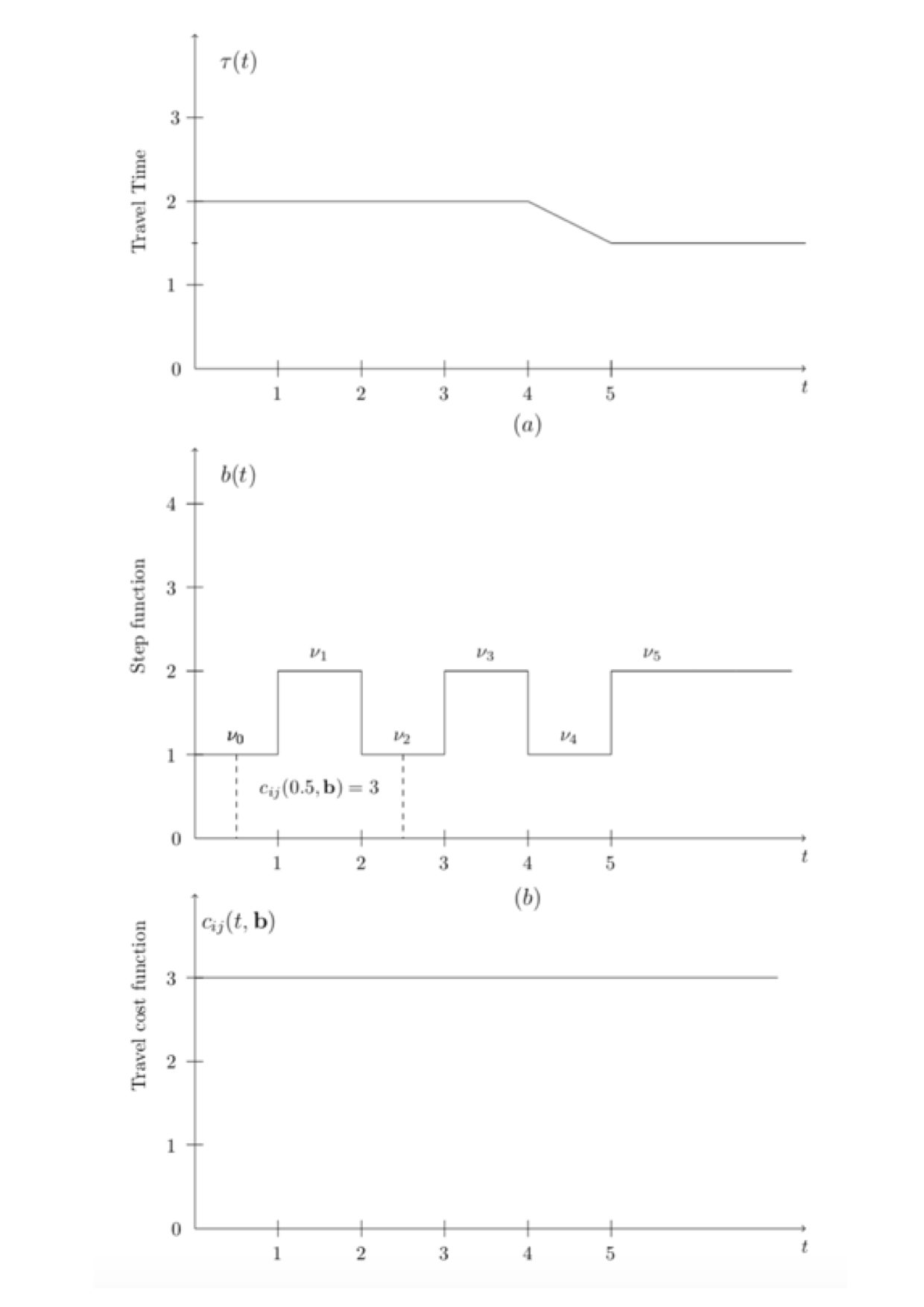}
\caption{A continuous piecewise linear arc travel time function $\tau_{ij}$, the associated constant step function $b(t)$.} \label{Figure_1}

\end{figure}

\subsection{Properties of travel cost function\textcolor{black}{s}}\label{sec_2_2}
We observe that, given a time-dependent graph $G=(V,A,\tau)$ and a step function $\mathbf{b}$, each output travel cost function $c_{ij}(t,\mathbf{b})$ is continuous piecewise linear function, with a number of breakpoints  not greater than $\mathcal{T}_{\tau}(i,j)+2\times |\mathcal{T}_b|$, with $(i,j)\in A$. In particular, we denote with $\mathcal{T}_{c}(i,j,\mathcal{T}_b)$ the corresponding set of breakpoints. From Proposition \ref{prop_rank}, it descends that we are interested in determining if there exists a step function $\mathbf{b}^*$ satisfying (\ref{constant_c}):
\begin{equation}\label{constant_c}
\mathcal{T}_c(i,j,\mathcal{T}_{b^*})=\emptyset.
\end{equation}
For this reason, we now provide both sufficient conditions and necessary conditions for asserting that a time instant $t$ is a breakpoint of a travel cost function $c_{ij}(t,\textbf{b})$, with $(i,j)\in A$ and $t\geq 0$.\\
\indent Let $\Gamma_{ij}(t)$ denote the arrival time at node $j$ when the vehicle starts to traverse the arc $(i,j)\in A$ at time instant $t$, i.e. $\Gamma_{ij}(t)=t+\tau_{ij}(t)$. First from (\ref{c_def}) we have that if a time instant $t$ is a breakpoint of $c_{ij}(t,\textbf{b})$ then at least one of the following necessary conditions hold true: $t$ is a breakpoint of $b(t)$; $\Gamma_{ij}(t)$  is a breakpoint of $b(t)$; $t$ is a breakpoint of $\tau_{ij}(t)$. More formally:
\begin{equation}\label{t_c}
t\in\mathcal{T}_{c}(i,j,\mathcal{T}_b)\quad\Rightarrow \quad t\in \mathcal{T}_b \vee t\in\mathcal{T}_{\tau} (i,j ) \vee\Gamma_{ij}(t)\in \mathcal{T}_b,
\end{equation}

As far as  sufficient conditions is concerned,  a time instant $t$ is a breakpoint of $c_{ij}(t,\textbf{b})$ if the following three conditions hold:  $t $ is  a breakpoint of the step function $\mathbf{b}$; $t$ is not a breakpoint of the travel time function $\tau_{ij}(t)$;  the arrival time $\Gamma_{ij}(t)$ is not a breakpoint of $\mathbf{b}$. More formally:
\begin{equation}\label{suff_con}
t\in\mathcal{T}_b\wedge t\notin \mathcal{T}_\tau(i,j) \wedge \Gamma_{ij}(t)\notin\mathcal{T}_b\Rightarrow t\in\mathcal{T}_c(i,j,\mathcal{T}_b).
\end{equation}
The implication (\ref{suff_con}) has been  demonstrated in the Theorem \ref{Theo_bp_1}.
\textcolor{black}{
\begin{theorem}\label{Theo_bp_1}
Let us suppose that we are given an arc $(i,j)\in A$, a step function $b(t)$ and one of its breakpoint $t_p$, that is $t_p\in\mathcal{T}_b$ with $p=1,\dots,|\mathcal{T}_b|$. If  both the conditions (a) and (b) hold true, then  a breakpoint of $c_{ij}(t,\textbf{b})$ also occurs at time instant $t_p$.
 \begin{enumerate}[(a)]
\item The time instant $t_p$ is not a breakpoint of $\tau_{ij}(t)$, i.e. $t_p\notin\mathcal{T}_\tau(i,j)$.
\item  The arrival time  $\Gamma_{ij}(t_p)$ is not  a breakpoint of $b(t)$, that is $\Gamma_{ij}(t_p)\notin\mathcal{T}_b$
\end{enumerate}
\end{theorem}
 \begin{proof}
 If we write Equation (\ref{c_def}) for the time instant  $t_p$,  we obtain:
 \begin{equation}\label{c_def_2}
c_{ij}(t_p,\mathbf{b})=(t_{p+1}-t_p)b_{p}+\sum_{\ell=p+1}^{q-1} (t_{\ell+1}-t_{\ell})b_{\ell}+(t_p+\tau_{ij}(t_p)-t_{q})b_{q}.
\end{equation}
The thesis is proved if we demonstrate that there exists $\Delta>0$ , such that 
$$\frac{c_{ij}(t_p,\mathbf{b})-c_{ij}(t_p-\Delta,\mathbf{b})}{\Delta}\neq\frac{c_{ij}(t_p+\Delta,\mathbf{b})-c_{ij}(t_p,\mathbf{b})}{\Delta}.$$
 From the hypothesis it results that there exists a $\Delta>0$ such that the following conditions hold true. (1) The travel time function $\tau_{ij}(t)$ does not change its slope in the time interval $[t_p-\Delta,t_p+\Delta]$. (2) The  step function $b(t)$ does not change its value during the time intervals  $[(t_p-\Delta),t_p]$, $[t_p,(t_{p}+\Delta)]$ and $[\Gamma_{ij}(t_p-\Delta),\Gamma_{ij}(t_p+\Delta)]$. This implies that:
$$[(t_p-\Delta),t_p]\subseteq[t_{p-1},t_p]\wedge[t_p,(t_p+\Delta)]\subseteq[t_{p},t_{p+1}]$$
and
$$[\Gamma_{ij}(t_p-\Delta),\Gamma_{ij}(t_p+\Delta)]\subseteq[t_{q},t_{q+1}], $$
 with  $q>p$.
Let us  write Equation (\ref{c_def}) for the time instant $(t_p-\Delta)$:
\begin{equation}\label{c_def_1}
c_{ij}(t_p-\Delta,\mathbf{b})=(t_{p}-t_p+\Delta)b_{p-1}+\sum_{\ell=p}^{q-1} (t_{\ell+1}-t_{\ell})b_{\ell}+(t_p-\Delta+\tau_{ij}(t_p-\Delta)-t_{q})b_{q}.
\end{equation}
By subtracting (\ref{c_def_1}) from (\ref{c_def_2}), we obtain: 
\begin{equation}\label{quad4}
\frac{c_{ij}(t_p,\mathbf{b})-c_{ij}(t_p-\Delta,\mathbf{b})}{\Delta}=-b_{p-1}+ b_{q}+\frac{ \tau_{ij}(t_p)-\tau_{ij}(t_p-\Delta)}{\Delta}b_{q}.
\end{equation}
Similarly let us subtract (\ref{c_def_2}) from (\ref{c_def}) rewritten for time instant  $(t_p+\Delta)$
\begin{equation}\label{quad4.1}
\frac{c_{ij}(t_p+\Delta,\mathbf{b})-c_{ij}(t_p,\mathbf{b})}{\Delta}=-b_{p}+ b_{q}+\frac{ \tau_{ij}(t_p+\Delta)-\tau_{ij}(t_p)}{\Delta}b_{q}\nonumber
\end{equation}
Since $t_p$ is not a breakpoint of $\tau_{ij}(t)$, then we have that:
$$\frac{c_{ij}(t_p,\mathbf{b})-c_{ij}(t_p-\Delta,\mathbf{b})}{\Delta}-\frac{c_{ij}(t_p+\Delta,\mathbf{b})-c_{ij}(t_p,\mathbf{b})}{\Delta}=-b_{p-1}+b_p.$$
As $t_p$ is a breakpoint of the given step function we have that $b_{p-1}\neq b_p$. The thesis is proved.\\
\end{proof}}

\section{The \textcolor{black}{Constant Traversal Cost Problem}} \label{sec_3} 
As stated by Proposition \ref{prop_rank}, the membership of a graph $G$ in the set $\mathcal{L}_C$ implies its path ranking invariance. To ease the discussion, from now on, we refer to the decision problem corresponding to the set $\mathcal{L}_C$ as the \textit{\textcolor{black}{Constant Traversal Cost Problem}}.
\begin{problem}
  \problemtitle{\textbf{\textcolor{black}{Constant Traversal Cost Problem}}}
  \probleminput{A time-dependent graph $G$.}
  \problemquestion{Is there a step  function $\mathbf{b}^*$  such that the corresponding travel cost $c_{ij}(t,\textbf{b}^*)$ associated to each arc $(i,j)\in A$ is constant?}
\end{problem}
In this section, we demonstrate that the \textcolor{black}{CTCP} is decidable. This goal is gained by demonstrating that there exists a \textit{computable function} $\mathbf{1}_C(G)$ which is characteristic for the set $\mathcal{L}_C$ of time-dependent graphs,  that is $\mathbf{1}_C(G):=1$ if $G\in\mathcal{L}_C$ and  $\mathbf{1}_C(G):=0$ if $G\notin\mathcal{L}_C$.

\subsection{Definition of the characteristic function}\label{char_func}
We denote with $\gamma(\mathcal{T}_b,\textbf{b})$ the \textit{maximum travel cost range} of a graph $G$ with respect to a step cost function $\textbf{b}$, defined as follows:
$$\gamma(\textbf{b}):=\max\limits_{(i,j)\in A}(\max\limits_{t\in[0,T]}c_{ij}(t,\textbf{b})-\min\limits_{t\in[0,T]}c_{ij}(t,\textbf{b})).$$ 
Let  denote with $\gamma^*$ the \textit{min-max travel cost range of a graph} $G$, that is:
\textcolor{black}{\begin{equation}\label{gamma_prob}
\gamma^*:=\min\limits_{\textbf{b}} (\gamma(\textbf{b})| b(t)>\rho\quad \forall t\geq0).
 \end{equation}}
By observing that a constant travel cost function has a range value zero, from  (\ref{gamma_prob}) it descends the following Proposition.
 \begin{prop}\label{omega}
The boolean function (\ref{1_C1}) is a total characteristic function of \textcolor{black}{CTCP}.
\begin{equation}\label{1_C1}
\mathbf{1}_{C}(G):=\begin{cases} 1& if\quad \gamma^*=0 \\0 \quad& otherwise  \end{cases}. 
\end{equation}
\end{prop}
\indent In order to assess the decidability of the \textcolor{black}{CTCP}, we have to devise an algorithm that  decides correctly whether the optimal value $\gamma^*$  has  value zero.  
We observe that the optimization problem (\ref{gamma_prob}) consists of two interdependent tasks: (a) determining the breakpoint set $\mathcal{T}_b$; (b) determining the values $(b_0,\dots,b_{|\mathcal{T}_b|-1})$ so that to minimize the corresponding maximum travel cost range.  As far as the first task is concerned we limit the corresponding search space to  
a finite and discrete set $\mathbf{\Omega}$, such that its elements are \textit{potential}  breakpoints of $\textbf{b}^*$, that is:
\begin{equation}\label{T__star}
t\in \mathcal{T}_{b^*}\Rightarrow t\in\mathbf{\Omega} .
\end{equation}
With the aim of determining a set $\mathbf{\Omega} $ satisfying these conditions, we associate to each arc $(i,j)\in A$ the set $\mathbf{\Omega}_{ij}$. In the definition of each $\mathbf{\Omega}_{ij}$ a key role is played by  ordered sets of time instants, termed \textit{time sequences generated by a given time instant on arc $(i,j)\in A$}.
\begin{definition}  \label{def_0} An ordered set of time instants $\mathbold{\omega}:=\{\omega_{1},\dots,\omega_{L}\}$, is a \textit{time sequence generated by time instant $\omega_{1}$ on the arc} $(i,j)\in A$ if the following conditions hold true:
$$\Gamma^{-1}_{ij}(\omega_1)<0,$$
$$\omega_\ell:=\Gamma_{ij}(\omega_{\ell-1})\quad \ell=2,\dots,L$$
$$\Gamma_{ij}(\omega_L)>T,$$
with  $0\leq \omega_{1}\leq \omega_{2}\leq \dots\leq \omega_{L}\leq T$
\end{definition}

We observe that there exists an infinite number of \textit{time sequences} that can be  generated on an arc $(i,j)\in A$: one for each time instant $t$ satisfying the first condition of Definition \ref{def_0}. For example for the travel time function in Figure 1 there exists one time sequence for each time instant belonging to the interval $[0,2[$. On the other hand, due to the FIFO property, the arrival time function $\Gamma_{ij}(t)$ is strictly increasing. This implies that each time instant $t\in [0,T]$ belongs to exactly one \textit{time sequence generated on the arc} $(i,j)\in A$. We reference such time sequence with $\mathbold{\omega}_{ij}(t):=\{\omega_{ij1}(t),\dots,\omega_{ijL}(t)\}$. In the numerical example of  Figure 1, both of the time instants $2.0$ and $4.0$ belong to the (unique) time sequence generated by $0.0$ on arc $(i,j)$, i.e.  $\mathbold{\omega}_{ij}(2.0):=\mathbold{\omega}_{ij}(4.0):=\{0.0,2.0,4.0\}$.
Given a time-dependent graph $G=(V,A,\tau)$, the set $\mathbf{\Omega}_{ij}$ is defined as the union set of the time sequences  $\mathbold{\omega}_{ij}(t)$, such that $t$ is a breakpoint of $\tau$, that is:
\begin{equation}\label{omegaij}
\mathbf{\Omega}_{ij}:=\bigcup\limits_{t\in\mathcal{T}}\mathbold{\omega}_{ij}(t),
\end{equation}

with $\mathcal{T}:=\bigcup\limits_{(i,j)\in A}\mathcal{T}_\tau(i,j)$ and $(i,j)\in A$. It is worth noting that since $\mathcal{T}$ is a subset of $\mathbf{\Omega}$, it is guaranteed that $\mathbf{\Omega}$ is not empty. 
\begin{theorem} \label{T1_C}
The set  $\mathbf{\Omega}$ defined in (\ref{Omega}) is a finite and discrete set of \textit{potential breakpoints} of $b^*(t)$.
\begin{equation}\label{Omega}
\mathbf{\Omega}:=\bigcap\limits_{(i,j)\in A}\mathbf{\Omega}_{ij}.
\end{equation}
\end{theorem}
\begin{proof}
We first observe that due to the FIFO property, given a time instant $t\in T$, the corresponding $\mathbold{\omega}_{ij}(t)$ is a finite and discrete set of time instants.
Therefore from  (\ref{omegaij}), (\ref{Omega}) and Definition \ref{def_0}, it results that $\mathbf{\Omega}$ consists of up to $|A|\times|\mathcal{T}|\times\frac{T}{\tau_{min}}$ where $\tau_{min}$ is the minimum traversal time  of $\tau$, that is $\tau_{min}:=\min\limits_{(i,j)\in A}(\tau_{ij}(t)|t\geq 0)$. This demonstrates that $\mathbf{\Omega}$ is a finite and discrete set. We now prove that (\ref{Omega}) defines a set of potential breakpoints for $\mathbf{b}^*$, that is:
\begin{equation}\label{cc}
t\in\mathcal{T}_{b^*}\Rightarrow t\in\mathbf{\Omega}_{ij},
\end{equation}
with $(i,j)\in A$.

We prove (\ref{cc}) by contradiction. Let us  suppose that there exists one arc $(i,j)\in A$ and a breakpoint $t^\prime$ of $b^*(t)$ such that $t'\notin\mathbf{\Omega}_{ij}$. We prove the thesis by demonstrating that:
\begin{equation}\label{thesis_1}
t\in\mathbold{\omega}_{ij}(t')\wedge t\in\mathcal{T}_{b^*}\Rightarrow t\in\mathcal{T}_c(i,j,\mathbf{b}^*),
\end{equation}
which contradicts the hypothesis that $t'$ is a breakpoint of a step function $\mathbf{b}^*$ generating a constant travel cost on arc $(i,j)$, that is $\mathcal{T}_c(i,j,\mathbf{b}^*)=\emptyset.$
We first observe that, since a time instant belongs to exactly one time sequence generated on arc $(i,j)$, then we have that $\mathbold{\omega}_{ij}(t')$ shares no time instant with $\mathbf{\Omega}_{ij}$. This implies that no time instant belonging to $\mathbold{\omega}_{ij}(t')$ is a breakpoint of the travel time function $\tau_{ij}(t)$, that is:
\begin{equation}\label{cond_aa}
t_\ell\in\mathbold{\omega}_{ij}(t')\Rightarrow t_\ell\notin\mathbf{\Omega}_{ij}\Rightarrow t_\ell\notin\mathcal{T}_\tau\subseteq\mathbf{\Omega}_{ij}.
\end{equation}
with $\ell=1,\dots,|\mathbold{\omega}_{ij}(t')|$.\\
From (\ref{suff_con}) and (\ref{cond_aa})  it results that (\ref{thesis_1}) can be demonstrate if we prove by induction on $\ell$ that:
\begin{equation}\label{thesis_ind}
t_\ell\in\mathbold{\omega}_{ij}(t')\wedge t_\ell\in\mathcal{T}_{b^*}\Rightarrow\Gamma_{ij}(t_\ell)\notin\mathcal{T}_{b^*}\wedge t_\ell\notin\mathcal{T}_\tau(i,j),
\end{equation}
with $\ell=1,\dots,|\mathbold{\omega}_{ij}(t')|$.\\
\noindent \underline{Case $\ell=|\mathbold{\omega}_{ij}(t')|$}. 
In this case $t_\ell$ denotes the last element of the ordered set  $\mathbold{\omega}_{ij}(t')$. From Definition \ref{def_0} it results that $\Gamma_{ij}(t_\ell)>T$.  Since any step function is constant in the long run, it results that 
the arrival time $\Gamma_{ij}(t_\ell)$ is not a breakpoint of  $\mathbold{b}^*$ for $\ell=|\mathbold{\omega}_{ij}(t')|$. From (\ref{cond_aa}) it results that (\ref{thesis_ind}) holds true for $\ell=|\mathbold{\omega}_{ij}(t')|$.\\
\underline{Case $\ell\leq|\mathbold{\omega}_{ij}(t')|-1$}. We suppose by induction that $t_{\ell+1}$ satisfies (\ref{thesis_ind}). Since  $t_{\ell}$ is the predecessor of $t_{\ell+1}$ in $\mathbold{\omega}_{ij}(t')$,  then from Definition \ref{def_0} 
we have that $t_{\ell+1}=\Gamma_{ij}(t_\ell)$. From (\ref{suff_con}), (\ref{constant_c})  and the induction hypothesis, it descends that $\Gamma_{ij}(t_\ell)$ cannot be a potential breakpoint of $b^*(t)$, i.e. $\Gamma_{ij}(t_\ell)\notin\mathcal{T}_{b^*}$. 
 From (\ref{cond_aa}) it results that (\ref{thesis_ind}) holds true for $\ell<|\mathbold{\omega}_{ij}(t')|$.\ \\

\end{proof}
\textcolor{black}{Given a constant value $\rho>0$, let us define the following restriction of (\ref{gamma_prob}):
\begin{equation}\label{gamma_prob1}
\overline{\gamma}:=\min\limits_{\textbf{b}}( \gamma(\textbf{b})|\mathcal{T}_b\subseteq\mathbf{\Omega}\wedge b(t)\geq\rho\quad \forall t\geq0).
 \end{equation}
}
\begin{prop}\label{ub}
Given a time-dependent graph $G$,  the  travel cost range $\overline{\gamma}$ is equal to  zero iff the $\gamma^*$ has value zero.
\end{prop}
\begin{proof}
Let us denote with  $\overline{\gamma}_1$ an upper bound of $\gamma^*$ defined as follows:
\textcolor{black}{$$\overline{\gamma}_1:=\min\limits_{\textbf{b}}( \gamma(\textbf{b})|\mathcal{T}_b\subseteq\mathbf{\Omega}\wedge b(t)>0\quad \forall t\geq0).$$}
We have that:
$$\gamma^*\leq\bar\gamma_1\leq\bar\gamma.$$
We divide the demonstration in two parts.\\
\underline{Part I}. First, we prove that $(\gamma^*=0\Leftrightarrow\overline{\gamma}_1=0)$.
Since $\overline{\gamma}_1$ is an upper bound of $\gamma^*$, then the necessity part is proved, that is  $\overline{\gamma}_1=0\Rightarrow\gamma^*=0$.
As far as the sufficiency part is concerned, we observe that the main implication of Theorem \ref{T1_C} is that $\mathcal{T}_{b^*}\subseteq\mathbf{\Omega}$. 
This implies that $\gamma^*=0\Rightarrow\overline{\gamma}_1=0$.\\
\underline{Part II}. We now prove that  $(\overline{\gamma}_1=0\Leftrightarrow\overline{\gamma}=0)$.
Since $\overline{\gamma}$ is an upper bound of $\overline{\gamma}_1$, then the necessity part is proved, that is  $\overline{\gamma}=0\Rightarrow\overline{\gamma}_1=0$.
Therefore, we only need to prove that, when $\overline{\gamma}_1$ is equal to zero,  it is always possible to determine a step cost function \textcolor{black}{$\overline{b}(t)\geq\rho$} generating on each arc $(i,j)\in A$ a constant travel cost function. 
Let us suppose that there exists a step cost function \textcolor{black}{$0<b^*(t)<\rho$}, generating  travel constant cost functions, that is 
$$c_{ij}(t, \mathbf{b}^*)=\underline{c}_{ij}(\mathbf{b}^*),$$
with $(i,j)\in A$ and $t\geq 0$. Since \textcolor{black}{$b^*(t)>0$}, then it is always possible to determine a positive number $\alpha>0$, such that:
$$\textcolor{black}{\alpha\times b^*(t)\geq\rho,}$$
with $t\geq0.$ If we set $\overline{b}(t)$ equal to $\alpha\times b^*(t)$, then we have that
\textcolor{black}{$$c_{ij}(t,\overline{\mathbf{b}})=\alpha\times\underline{c}_{ij}(\mathbf{b}^*),$$
with $t\geq0$. This implies that }
$$\gamma(\overline{\mathbf{b}},\mathcal{T}_{\overline{b}})=\alpha\times\gamma(\mathbf{b}^*,\mathcal{T}_{b^*})=0,$$
which proves the thesis.\\
\end{proof}

From Theorem \ref{T1_C} and Proposition \ref{ub} , it follows that $\mathbf{1}_{C}(G)$ can be reformulated as follows:
\begin{equation}\label{1_C3}
\mathbf{1}_{C}(G):=\begin{cases} 1& if\quad \overline{\gamma}=0 \\0 \quad& otherwise  \end{cases}. 
\end{equation}
As illustrated in the following section, the main advantage of formulation (\ref{1_C3}) is that the travel cost range $\overline{\gamma}$ is the optimal objective function value of a Linear Programming (LP) problem.
\subsection{Decidability of  the \textcolor{black}{CTCP}}

In order to assess the decidability of the \textcolor{black}{CTCP}, we outline an algorithm consisting of two steps. The former determines the set $\mathbf{\Omega}$ of potential breakpoints of $b^*(t)$. The latter  determines the optimal objective function value of  the optimization problem (\ref{gamma_prob1}). \\
\indent We recall that $\mathbf{\Omega}$ is the intersection of $|A|$ sets of time instants. 
Moreover each  $\mathbf{\Omega}_{ij}$ is the union of a finite number of time sequences $\mathbold{\omega}_{ij}(t)$, one for each travel time breakpoint $t\in\mathcal{T}$ and arc $(i,j)\in A$. Therefore, we only need to devise an iterative procedure that, given a time instant $t$ and a travel time function $\tau_{ij}(t)$, determines the ordered set $\mathbold{\omega}_{ij}(t)$ in a finite number of iteration steps. For this purpose we propose the Algorithm 1 consisting of two main steps.
\begin{enumerate}[(a)]
\item Firstly $\mathbold{\omega}_{ij}(t)$ is iteratively enriched by the arrival time $\Gamma_{ij}(t)$  associated to a start time equal to $t$, by the arrival time $\Gamma_{ij}(\Gamma_{ij}(t))$ associated to a start time equal to $\Gamma_{ij}(t)$, etc, until no time instant less than or equal to $T$ can be generated.
\item Finally, $\mathbold{\omega}_{ij}(t)$ is iteratively enriched by the start time $\Gamma_{ij}^{-1}(t)$  associated to an arrival time equal to $t$, by the start time $\Gamma_{ij}^{-1}(\Gamma_{ij}^{-1}(t))$ associated to an arrival time equal to $\Gamma_{ij}^{-1}(t)$, etc, until no time instant greater than or equal to $0$ can be generated.
\end{enumerate}
For example by running  Algorithm 1 iteratively for each distinct breakpoint of $\tau_{ij}(t)$ of Figure 1, we obtain the time sequences  $\mathbold{\omega}_{ij}(4.0)=\{0.0,2.0,4.0\}$  and $\mathbold{\omega}_{ij}(5.0)=\{1.0,3.0,5.0\}$. Therefore the set $\mathbf{\Omega}_{ij}$ consists of the time instants $\{0.0,1.0,2.0,$ $3.0,4.0,5.0\}.$ \\
\begin{prop}\label{Algo1}
Algorithm 1  converges in a finite number of iterations. 
\end{prop}
\begin{proof}
The thesis is proved by observing that each  while-condition evaluates to false after a finite number of iterations. Indeed, due to the FIFO hypothesis,  the arrival time function $\Gamma_{ij}(t)$ as well as the corresponding inverse function $\Gamma_{ij}^{-1}(t)$ are, respectively, strictly increasing and strictly decreasing functions, with $ (i,j)\in A$. 
\end{proof}

\begin{algorithm} \caption{\textcolor{black}{Determining the timed sequence $\mathbold{\omega}_{ij}(t)$ } }
\label{Alg1}
\begin{algorithmic}[h]
\State \textcolor{black}{INPUT: A continuous piecewise linear FIFO travel time function $\tau_{ij}(t)$ and a time instant $t$.} 
\State \textcolor{black}{OUTPUT: The  ordered set $\mathbold{\omega}$ of time instants. }

\State $\mathbold{\omega}\leftarrow\emptyset$
\State $\mathbold{\omega} \leftarrow t$
\State $t^{\prime}\leftarrow t$ 
\While {$(\Gamma_{ij}(t^{\prime})\leq T) \wedge ( \Gamma_{ij}(t^{\prime}) \notin \mathbold{\omega})$}
\State $\mathbold{\omega} \leftarrow \Gamma_{ij}(t^{\prime})$
\State $t^{\prime}\leftarrow  \Gamma_{ij}(t^{\prime})$ 
\EndWhile
\State $t^{\prime}\leftarrow t$ 
\While {$(\Gamma_{ij}^{-1}(t^{\prime})\geq 0) \wedge ( \Gamma_{ij}^{-1}(t^{\prime}) \notin \mathbold{\omega})$}
\State $\mathbold{\omega} \leftarrow \Gamma_{ij}^{-1}(t^{\prime})$
\State $t^{\prime}\leftarrow  \Gamma_{ij}^{-1}(t^{\prime})$ 
\EndWhile

\end{algorithmic} \end{algorithm}
Once it has been determined the set $\mathbf{\Omega}$,  it is possible to formulate an instance of  the linear program (\ref{obj})-(\ref{c7}).
A solution of such linear programming model  represents the parameters of a constant piecewise function $y(t)$ and a continuous piecewise linear function $x_{ij}(t)$, with $(i,j)\in A$. We partition the time horizon into a \textit{finite} number $|\mathbf{\Omega}|$ of time slots $[t_h,t_{h+1}] (h=0,\dots,|\mathbf{\Omega}|-1)$.
The  continuous variable $y_h$ represents the value of $y(t)$ during the $h-th$ time interval, that is:
$$y(t)=y_h,$$
with $t\in[t_h,t_{h+1}]$ and $h=0,\dots,|\mathbf{\Omega}|-1$. The function $x_{ij}(t)$ is the continuous piecewise linear function corresponding to the linear interpolation of the points $(t_{ijk},x_{ijk})$, that is :
$$x_{ij}(t_{ijk})=x_{ijk},$$
where $x_{ijk}$ is a  continuous variable, with $t_{ijk}\in\mathbf{\Omega}_{ij}$, $k=0,\dots,|\mathbf{\Omega}_{ij}-1|$ and $(i,j)\in A$.  The continuous variables $\underline{x}_{ij}$ and $\overline{x}_{ij}$ represent, respectively, the maximum and minimum value of $x_{ij}(t)$, with $(i,j)\in A$.
The continuous variable $\zeta_{ij}$ represents the range of the continuous piecewise linear function $x_{ij}(t)$, with $(i,j)\in A$. Finally the continuous variable $\zeta$ represents the max-min range value associated to the overall set of $\mathbf{x}$s functions generated by the problem. 

\begin{equation}\label{obj}
\zeta^*:=\min \zeta
\end{equation}
s.t.
{\allowdisplaybreaks
\begin{flalign}
\label{c1}
&x_{ijk}=\sum\limits_{h=0}^{|\mathbf{\Omega}|-1}a_{ijkh}\times y_h \quad\quad\quad\quad\quad\quad\quad\quad k=0,\dots,|\mathbf{\Omega}_{ij}|-1 ,\quad (i,j)\in A&& \\
\label{c2}
&\zeta\geq\overline{x}_{ij}-\underline{x}_{ij}\quad\quad\quad\quad\quad\quad\quad\quad\quad\quad\quad\quad  (i,j)\in A &&\\
\label{c3}
&\underline{x}_{ij}\leq x_{ijk}\quad\quad\quad\quad\quad\quad\quad\quad\quad\quad\quad\quad\quad k=0,\dots,|\mathbf{\Omega}_{ij}|-1,  \quad (i,j)\in A&&\\
\label{c4}
&\overline{x}_{ij}\geq x_{ijk}\quad\quad\quad\quad\quad\quad\quad\quad\quad\quad\quad\quad\quad k=0,\dots,|\mathbf{\Omega}_{ij}|-1 , \quad (i,j)\in A&&\\
\label{c5}
& y_{h}\geq\textcolor{black}{\rho} \quad\quad\quad\quad\quad\quad\quad\quad\quad\quad\quad\quad\quad\quad\quad h=0,\dots,|\mathbf{\Omega}|-1&&\\
\label{c6}
& x_{ijk}\geq0 \quad\quad\quad\quad\quad\quad\quad\quad\quad\quad\quad\quad\quad\quad k=0,\dots,|\mathbf{\Omega}_{ij}|-1 , \quad (i,j)\in A&&\\
\label{c7}
& \zeta\geq0 \quad\quad\quad\quad\quad\quad\quad\quad\quad\quad\quad\quad\quad &&
\end{flalign}
}

The objective function (\ref{obj}) states that the optimization model aims to determine a constant stepwise function $y^*(t)$, such that it is minimized the maximum range value of the corresponding $\mathbf{x}$s functions. Constraints (\ref{c1}) state the relationship between $y(t)$ and $x_{ij}(t)$ at time instant $t_{ijk}\in\mathbf{\Omega}_{ij}$. In particular each coefficient $a_{ijkh}$ represents the time spent on the arc $(i,j)\in A$ during period $[t_{h},t_{h+1}]$  if the start time is $t_{ijk}\in\mathbf{\Omega}_{ij}$, i.e.  
\begin{equation}\label{Asign1}
a_{ijkh}=\begin{cases}\min(t_{h+1}-t_{h},\max(0,\Gamma_{ij}(t_{ijk})-t_{h})) \quad& if\quad k\leq h \\0 \quad& if \quad k>h \end{cases},
\end{equation}
with $h=0,\dots,|\mathbf{\Omega}|-1$, $k=0,\dots,|\mathbf{\Omega}_{ij}|-1$.
Constraints (\ref{c2}) state the relationship between the objective function $\zeta$  and the range value of $x_{ij}(t)$, modeled as the difference between $\overline{x}_{ij}$ and  $\underline{x}_{ij}$. Constraints (\ref{c3}) and  (\ref{c4}) state the relationship between $\underline{x}_{ij}$, $\overline{x}_{ij}$ and the continuous variables $x_{ijk}$. Constraints (\ref{c5}) state that the  \textcolor{black}{constant stepwise} linear function \textcolor{black}{ $y(t)$ has to be greater  or equal than the input parameter $\rho>0$}. Constraints (\ref{c6}) and (\ref{c7}) provide the non-negative conditions of the remaining decision variables. 
\begin{theorem}\label{cert_feas}
Given a time-dependent graph $G$,  the optimal solution of the linear program (\ref{obj})-(\ref{c7}) is also optimal for the optimization problem (\ref{gamma_prob1}).
\end{theorem}
\begin{proof}
Let $(\mathcal{T}_{\mathbf{y}^*},\mathbf{y}^*)$ denote the parameters of the constant stepwise function $y^*(t)$ optimal for the linear program (\ref{obj})-(\ref{c7}). We observe that  $y^*(t)$ can be used \textit{to generate} two continuous piecewise linear functions for each arc $(i,j)\in A$. The former is the travel cost function $c_{ij}(t,\mathbf{y}^*)$. The latter is  the  continuous piecewise linear $x^*_{ij}(t)$ defined as the linear interpolation associated to the optimal values  $x^*_{ijk}$ with $k=0,\dots|\mathbf{\Omega}_{ij}|-1$ and $(i,j)\in A$. We want to prove that $c_{ij}(t,\mathbf{y}^*)=x^*_{ij}(t)$ for $t\geq 0$. We start by observing that the right hand side of  constraints  (\ref{c1}) corresponds to the right-hand side of (\ref{c_def}) when $t=t_{ijk}\in \mathbf{\Omega}_{ij}$, that is $c_{ij}(t_{ijk},\mathbf{y}^*)=x^*_{ij}(t_{ijk})$, with $(i,j)\in A$ and $k=0,\dots|\mathbf{\Omega}_{ij}-1|$. Therefore, the thesis is proved if we demonstrate that  all breakpoints of each travel cost function  $c_{ij}(t,\mathbf{y}^*)$ belong to $\mathbf{\Omega}_{ij}$, with $(i,j)\in A$. 
From (\ref{t_c}) we have that:
\begin{equation}\label{c11}
 t\in\mathcal{T}_{c}(i,j,\mathbf{y}^*)\quad\Rightarrow \quad t\in \mathcal{T}_{y^*} \vee t\in\mathcal{T}_{\tau} (i,j ) \vee\Gamma_{ij}(t)\in \mathcal{T}_{y^*}.
 \end{equation}
Since there is no constraint stating that $y^*_h\neq y^*_{h+1}$, then  
the breakpoints set $\mathcal{T}_{y^*}$ is a subset of each $\mathbf{\Omega}_{ij}$, with $(i,j)\in A$. Moreover, from the definition of $\mathbf{\Omega}_{ij}$, we have that all breakpoints of $\tau_{ij}(t)$ belongs to it. Finally we observe that from the definition of $\mathbf{\Omega}_{ij}$, it descends that:
$$\Gamma_{ij}(t)\in \mathcal{T}_{y^*}\subseteq\mathbf{\Omega}_{ij}\Rightarrow \mathbold{\omega}_{ij}(\Gamma_{ij}(t))\subseteq\mathbf{\Omega}_{ij}\Rightarrow t\in \mathbf{\Omega}_{ij}$$
Therefore, the implication (\ref{c11}) can be rewritten as:
$$t\in\mathcal{T}_c(i,j,\mathcal{T}_{y^*})\Rightarrow t\in\mathbf{\Omega}_{ij},$$
which proves the thesis.
\end{proof}
From Proposition \ref{Algo1} and  Theorem \ref{cert_feas} it descends that, given a time-dependent graph $G=(V,A,\tau)$, the value of the characteristic function  $\mathbf{1}_C(G)$ can be computed as follows. \\
\underline{Step 1}]. Determine $\mathbf{\Omega}$ by iteratively running Algorithm 1, for each arc $(i,j)\in A$ and each breakpoint of the  travel time function $\tau_{ij}(t)$.\\
\underline{Step 2}. Solve the linear program (\ref{obj})-(\ref{c7}). The value the function $\mathbf{1}_C(G)$ takes on as output  is 
$$\mathbf{1}_{C}(G):=\begin{cases} 1& if \quad \zeta^*=0 \\0 \quad&otherwise  \end{cases}. $$

\section {A lower bounding procedure}\label{sec_4}
If the \textcolor{black}{CTCP} feasibility check fails, then the optimal solution of the linear program  (\ref{obj})-(\ref{c7}) can be used to determine  "good" lower bounds for the optimal solutions of the class of time-dependent routing problems (\ref{opt_ref}). \\
\indent We preliminarily define a parameterized family of travel time functions $\underline{\tau}(\textbf{b})$, where the parameter is the step function $\textbf{b}$. In particular, given a time-dependent graph $G=(V,A,\tau)$,  the travel time function $\underline{\tau}_{ij}(t,\textbf{b})$ we generate is a \textit{lower approximation} of the \textit{original} travel time function $\tau_{ij}(t)$, that is:
$$\underline{\tau}_{ij}(t,\textbf{b})\leq\tau_{ij}(t),$$
for each $t\geq 0$ and $(i,j)\in A$. We start by observing that, given a step function $\textbf{b}$, it  is univocally associated to each arc $(i,j)\in A$ the travel cost function $c_{ij}(t,\textbf{b})$. In order to generate $\underline{\tau}_{ij}(t, \textbf{b})$, we first \textit{approximate} such travel cost function to its minimum value $\underline{c}_{ij}(\textbf{b})=\min\limits_{t}c_{ij}(t,\textbf{b})$.  Then $\underline{\tau}_{ij}(t,\textbf{b})$ is defined as the output function of the travel time model proposed by \cite{Ichoua2003} (IGP model).
\begin{definition}
The travel time function $\underline{\tau}_{ij}(t,\textbf{b})$ is the continuous piecewise linear function generated by IGP model, where the input parameters are the step function $\textbf{b}$ and the constant value  $\underline{c}_{ij}(\textbf{b})$.
\end{definition} 
According to the IGP model, given a start  time $t$ the travel time value $\underline{\tau}_{ij}(t,\textbf{b})$ is computed by an iterative procedure (Algorithm \ref{alg:alg2}). The relationship between the  input parameters  and the output value of Algorithm \ref{alg:alg2} can be expressed in a compact fashion as follows:
\begin{equation}\label{IGP}
\underline{c}_{ij}(\textbf{b})=\int_{t}^{t+\underline{\tau}_{ij}(t,\textbf{b})}b(\mu)d\mu.
\end{equation}
\begin{algorithm}[H]
\caption{Computing the travel time  $\underline{\tau}_{ij}(t,\textbf{b})$}
\label{alg:alg2}

\begin{algorithmic}
\State INPUT: A step function $\mathbf{b}=[(t_0,b_0),\dots,(t_{|\mathcal{T}_b-1|},b_{|\mathcal{T}_b-1|})]$, a constant cost $\underline{c}_{ij}(\textbf{b})$  and the start time $t$
\State

\State $k \leftarrow p: t_{p} \le t \le t_{p+1}$

\State $d \leftarrow\underline{c}_{ij}(\textbf{b})$

\State $t' \leftarrow t + d / b_{p}$

\While{$t' > t_{k+1}$}

\State $d \leftarrow d - b_{k}(t_{k+1}-t)$ 

\State $t \leftarrow t_{k+1}$

\State $t^{\prime} \leftarrow t + d/b_{k+1}$

\State $k \leftarrow k+1$

\EndWhile
\Return $t' - t$

\end{algorithmic}

\end{algorithm}
\begin{prop}\label{low_approx}
The travel time function $\underline{\tau}_{ij}(t,\textbf{b})$ is a lower approximation of the original travel time function $\tau$.
\end{prop}
\begin{proof}
We observe that a compact formulation of (\ref{c_def}) is the following:
$$c_{ij}(t,\textbf{b})=\int_{t}^{t+\tau_{ij}(t)}b(\mu)d\mu.$$
From (\ref{IGP}) and the definition of $\underline{c}_{ij}(\textbf{b})$, it results that:
$$\underline{c}_{ij}(\textbf{b})\leq c_{ij}(t,\textbf{b})\Leftrightarrow \underline{\tau}_{ij}(t,\textbf{b})\leq\tau_{ij}(t),$$
with $t\geq 0$, which proves the thesis.
\end{proof}
\begin{prop}\label{G_b_rank_inv}
Given a time-dependent graph $G$ and a step function $\textbf{b}$, the time depedent graph $\underline{G}_{\textbf{b}}=(V,A,\underline{\tau}(\textbf{b}))$ is path ranking invariant.
\end{prop}
\begin{proof}
Since $\textbf{b}$ is a step function, from (\ref{IGP}) it results that if the right-hand side of (\ref{c_def}) is evaluated w.r.t $\underline{\tau}_{ij}(t,\textbf{b})$, then we obtain the constant value $\underline{c}_{ij}(\textbf{b})$, for any time instant $t\geq 0$ and arc $(i,j)\in A$. This implies that the time-dependent graph $\underline{G}_{\textbf{b}}=(V,A,\underline{\tau}(\textbf{b}))$ is a yes instance of the \textit{\textcolor{black}{Constant Traversal Cost Problem}}, which proves the thesis. 
\end{proof}
The family of lower approximations $\underline{\tau}(\textbf{b})$ gives rise to a parameterized family of lower bounds $\underline{z}(\textbf{b})$, defined as follows. 
We denote with $\underline{z}(p_k,t,\textbf{b})$ the traversal time of a path $p_k$ at time instant $t$ on the time-dependent graph  $\underline{G}_{\textbf{b}}=(V,A,\underline{\tau}(\textbf{b}))$, that is 
\begin{equation}\label{z_tau_1}
\underline{z}(p_k,t, \textbf{b})=\underline{z}(p_{k-1},t, \textbf{b})+\underline{\tau}_{i_{k-1}i_k}(\underline{z}(p_{k-1},t), \textbf{b}),
\end{equation}
with the initialization $\underline{z}(p_0,t, \textbf{b})=0$.  
Since $\underline{G}_{\textbf{b}}$ is \textit{less congested} than $G$, the duration of a path is shorter on $\underline{G}_{\textbf{b}}$, that is $\underline{z}(p_k,t, \textbf{b})\leq z(p_k,t)$. Given a set $\mathcal{P}$ of paths defined on $(V,A)$, we compute the lower bound $\underline{z}(\textbf{b})$ as the duration of the quickest path of $\mathcal{P}$ on $\underline{G}_{\textbf{b}}$, that is:
$$\underline{z}(\textbf{b})=\min\limits_{p\in\mathcal{P}}\underline{z}(p,t,\textbf{b})\leq \min\limits_{p\in\mathcal{P}}z(p,t).$$ 
The main implication of  Proposition \ref{G_b_rank_inv} is that the lower bound $\underline{z}(\textbf{b})$ can be computed by solving a time-invariant routing problem.  
 In particular,  we have that:
$$\underline{z}(\textbf{b})=\underline{z}(\underline{p}^*_{\textbf{b}},t,\textbf{b}),$$
where
$$\underline{p}^*_{\textbf{b}}=\arg\min\limits_{p\in\mathcal{P}}\sum\limits_{(i,j)\in p}\underline{c}_{ij}(\textbf{b}).$$\\
In determining $\underline{p}^*_{\textbf{b}}$ we exploit that a cost function \textit{valid} for $\underline{G}_{\textbf{b}}$ is $d(i,j)=\underline{c}_{ij}$, with $(i,j)\in A$. Indeed, since $\underline{\tau}_{ij}(t,\textbf{b})$ is constant in the long run, we have that $\underline{\tau}_{ij}(T,\textbf{b})=\underline{c}_{ij}/b_{\mathbf{|\mathcal{T}_b|-1}}$. \\
\indent In order to find the best (larger) lower bound, the following problem has to be solved:
\begin{equation}\label{best_lb}
\max\limits_{\textbf{b}}\underline{z}(\textbf{b}).
\end{equation}
Unfortunately, this problem is nonlinear, nonconvex and non differentiable. So it is quite unlikely to determine its optimal solution with a moderate computational effort.
Instead, we aim to find a good lower bound as follows. 
We  observe that the tightness of the lower bound $\underline{z}(\textbf{b})$ clearly depends on the \textit{maximum fitting deviation} between the original travel time function $\tau$ and its lower approximation $\underline{\tau}(\textbf{b})$. It is worth noting that we generate $\underline{\tau}_{ij}(\textbf{b})$ by approximating  travel cost function $c_{ij}(t,\textbf{b})$ to its minimum value $\underline{c}_{ij}(\textbf{b})$, with $(i,j)\in A$.
The maximum fitting deviation of such travel cost approximation is the maximum range value $\gamma(\textbf{b})$ defined in Section \ref{char_func}.  If $\gamma(\textbf{b})$ is equal to zero, then both travel time function $\underline{\tau}(\textbf{b})$ and each travel cost $\underline{c}_{ij}(\textbf{b})$ are \textit {perfect fit}, with $(i,j)\in A$. In this case, the original graph $G$ is path ranking invariant and $\underline{z}(\textbf{b})$ is the best (larger) lower bound. Otherwise the value of $\gamma({\textbf{b}})$ represents a measurement of the \textit{distance} from this special case.\\ 
\indent For these reasons we heuristically solve (\ref{best_lb}), by determining the optimal step function $\textbf{y}^*$ of  the linear program (\ref{obj})-(\ref{c7}), which represents the step function with the min-max value of travel cost range.
In particular\textcolor{black}{,} we observe that the coefficient $\underline{c}_{ij}(\textbf{y}^*)$ are also determined by the optimal solution of such linear program. Indeed the optimal piecewise linear function $x^*_{ij}(t)$ corresponds to the travel cost function $c_{ij}(t,\textbf{y}^*)$, and, therefore, its minimum value corresponds the optimal value of the decision variable $\underline{x}_{ij}$, namely $\underline{x}^*_{ij}$ with $(i,j)\in A$.
Summing up the proposed lower bounding procedure consists of  three main steps. 
\begin{itemize}
\item \textbf{STEP 1}. Solve the linear program (\ref{obj})-(\ref{c7}). Set the step function $\textbf{b}$ equal to $\textbf{y}^*$. Similarly we set the coefficient $\underline{c}_{ij}$ to  $\underline{x}^*_{ij}$ for each $(i,j)\in A$. 
\item \textbf{STEP 2}. Determine $\underline{p}_{\textbf{b}}^*$ as the least cost solution of the following time-independent routing problem:
$$\min\limits_{p\in \mathcal{P}}\sum\limits_{(i,j)\in p}\underline{c}_{ij}(\textbf{b})$$
\item  \textbf{STEP 3}. Compute the lower bound $\underline{z}(\textbf{b})$ by evaluating $\underline{p}_{\textbf{b}}^*$ w.r.t. $\underline{\tau}(\textbf{b})$, that is:
$$\underline{z}(\textbf{b})=\underline{z}(\underline{p}_{\textbf{b}}^*,t_0,\textbf{b})$$
\end{itemize}

It is worth noting that due to a huge number of travel time breakpoints, the linear program (\ref{obj})-(\ref{c7}) might not be solved in a reasonable amount of time (e.g., in realistic time-dependent graph). In this cases we update the first step of the lower bounding procedure as follows. We solve a smaller instance of the linear program, where  each $\mathbf{\Omega}_{ij}$ is set equal to a given (unique) set of time instant $\mathcal{B}$. 

It results that each optimal continuous piecewise linear function $x^*_{ij}(t)$ is a surrogate function of the travel cost  $c_{ij}(t,\textbf{y}^*)$. This implies that during the first step we compute the coefficient $\underline{c}_{ij}(\textbf{y}^*)$ by enumerating all breakpoints  of $c_{ij}(t,\textbf{y}^*)$, $(i,j)\in A$.\\
\indent We finally observe that since the path $\underline{p}^*_{\textbf{b}}$ belongs to the set of feasible paths $\mathcal{P}$,  we also generate a parameterized  family of upper bound $\overline{z}(\textbf{b})$ obtained by evaluating $\underline{p}^*_{\textbf{b}}$ w.r.t. the original travel time function $\tau$ :
$$\overline{z}(\textbf{b}):=z(\underline{p}^*_{\textbf{b}},t_0).$$

\section{Computational Results}\label{sec_5}
\indent We have tested our lower bounding  procedure on two well-known  vehicle routing problems: the \textit{Time Dependent Travelling Salesman Problem}  and the \textit{Time Dependent Rural Postman Problem}. \\
 The state-of-the-art algorithms for such time-dependent routing problems are, to the best of our knowledge, \cite{ADAMO2020104795}(TDTSP) and \cite{calogiuri2019branch} ($TDRPP$). Both  contributions \textcolor{black}{gave rise to} lower bounding procedures which \textcolor{black}{were} used in an enumerative scheme. It is worth noting that \textit{the idea of a lower bound based on a less congested graph} $\underline{G}=(V,A,\underline{\tau})$ is also the idea underlying the lower bounds proposed by \cite{ADAMO2020104795} and \cite{calogiuri2019branch}. 
Figure \ref{Figure_2} reports an example concerning a time-dependent travel time function $\tau_{ij}(t)$ associated with an arc of a time-dependent graph $G$, consisting of both time-invariant arcs and time-dependent arcs. The Figure also shows the $\underline{\tau}_{ij}(t)$ obtained by our approach along with the lower approximations of $\tau_{ij}(t)$ obtained by applying the approach proposed, respectively, by \cite{calogiuri2019branch} and \cite{ADAMO2020104795}. 
In particular, \cite{calogiuri2019branch} assumed  that the \textit{original} traversal time $\tau_{ij}(t)$ was generated from the model by \cite{Ichoua2003379} as the time needed to traverse an arc $(i,j)\in A$ of length $L_{ij}$ at step-wise speed $v_{ij}(t)$, when the vehicle leaves the vertex $i$ at time $t$. Then \cite{calogiuri2019branch} exploited the speed decomposition proposed by \cite{Cordeau2014}, where given a partition of the planning horizon in $H$ time periods,  the speed value $v_{ijh}$ of arc $(i,j)\in A$ during period $h$ is express as:
$$v_{ijh}=\delta_{ijh}f_hu_{ij},$$
where 
\begin{itemize}
\item $u_{ij}$ is the maximum travel speed across arc $(i,j)\in A$, i.e., $u_{ij}=\max\limits_{h=0,\dots,H-1}v_{ ijh}$;
\item $f_h$ belongs to$ ]0,1]$ and is the best (i.e., lightest) congestion factor during $h-th$ interval, i.e., $f_h =\max\limits_{(i,j)\in A} v_{ijh} /u_{ij}$ ;
\item $\delta_{ijh}$ belongs to $]0,1]$ and represents the degradation of the congestion factor of arc $(i,j)$ in the $h-th$ interval with respect to the less-congested arc in the same period.
\end{itemize}
In  \cite{calogiuri2019branch}   the lower approximation $\underline{\tau}$ was defined as the traversal time of the vehicle evaluated w.r.t.  the most favourable congestion factor during each interval $h$-th, i.e., $\underline{v}_{ijh}.\leftarrow f_hu_{ij}$. The main drawback of this approach is that, if a subset of arcs are time-invariant, then traffic congestion factors $f_h$ are all equal to 1. In these cases, the lower approximation $\underline{\tau}$ determined by  \cite{calogiuri2019branch} are constant and equal to $L_{ij}/u_{ij}$ (see Figure \ref{Figure_2}). \cite{ADAMO2020104795} enhanced the speed decomposition of \cite{Cordeau2014}, by determining its parameters through a fitting procedure of the traffic data. As shown in the example reported in Figure \ref{Figure_2}, by applying our approach, we get  a lower approximation $\underline{\tau}$ that fits the original $\tau$ better than the travel time approximation determined by  \cite{ADAMO2020104795}. 
Indeed, the fitting procedure by \cite{ADAMO2020104795} aims to minimize the deviation between the (original) speed values $v_{ijh}$ and the most favourable speed value $\underline{v}_{ijh}$, during some (not necessarily consecutive) periods. As shown in Figure \ref{Figure_2} this approach do not assure that there exists at least one time instant closing the fitting deviation between the lower approximation and the original travel time $\tau$, which, on the contrary, it is guaranteed by our procedure.\\
As discussed in the following subsections, computational results show that tighter approximation travel times implies tighter lower bounds for both TDTSP instances and TDRPP instances. In particular, when embedded into a branch-and-bound procedure, the proposed lower bounding mechanism allows to solve to optimality a larger number of instances than the state-of-the-art algorithms.
\begin{figure}[H] 
\begin{center}
 \includegraphics[width=110mm,  angle =-90]{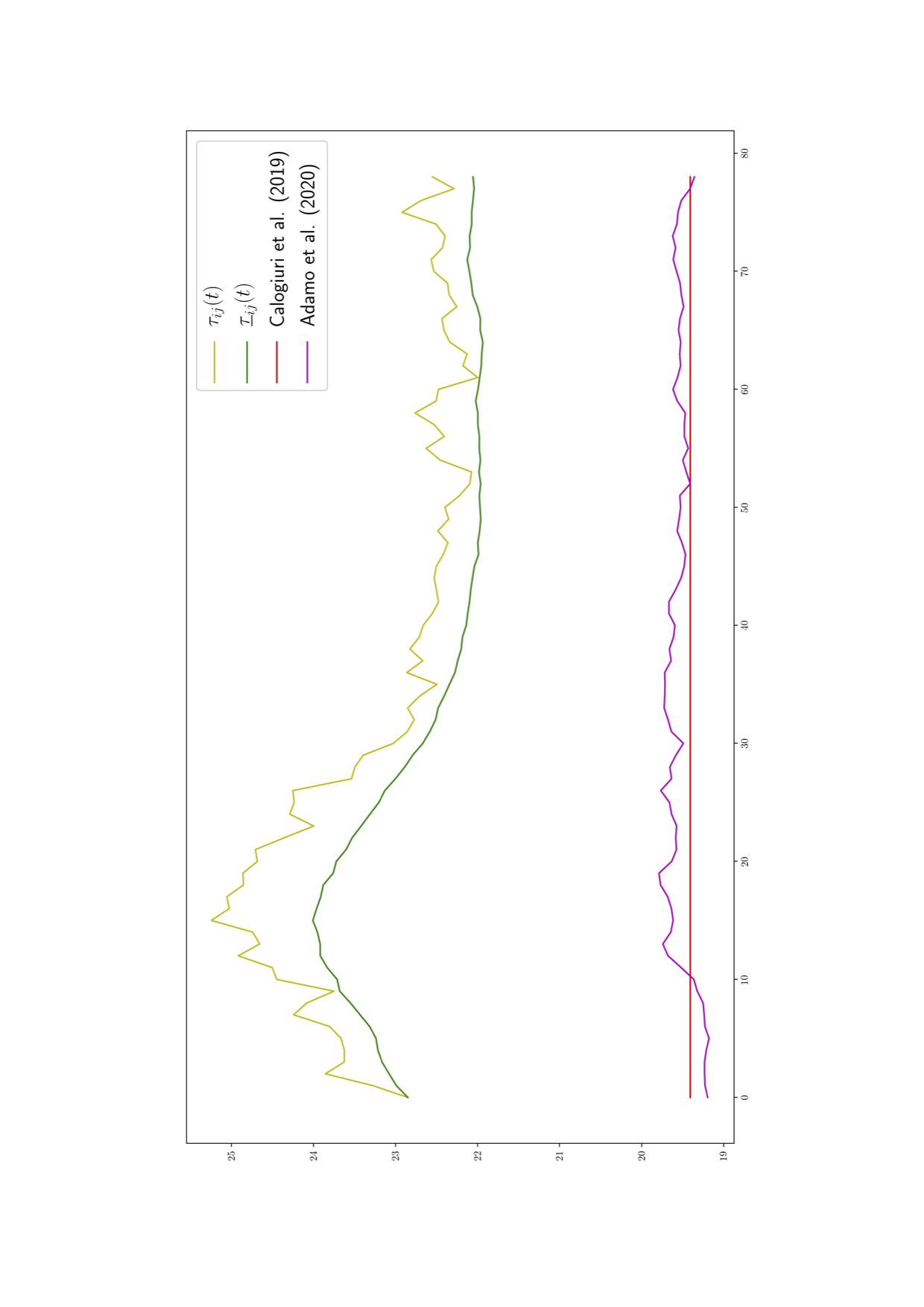}
\caption{Comparing the $\underline{\tau}$ functions determined by, respectively, our lower bounding procedure,  \cite{Cordeau2014},  \cite{ADAMO2020104795} } \label{Figure_2}
\end{center}
\end{figure} 
 
\subsection {Results on the TDTSP}
The  procedure illustrated in Section \ref{sec_4} can be used to determine a "good" lower bound for $TDTSP$ as follows. We run the first step of the proposed lower bounding procedure and solve the linear program  (\ref{obj})-(\ref{c7}).  Then we run the second and third steps of the lower bounding procedure. In particular the solution $\underline{p}^*_{\textbf{b}}$ corresponds to a first hamiltonian circuit determined by solving an Asymmetric TSP with arc costs equal to $d_{ij}=\min\limits_{t\geq 0}x^*_{ij}(t).$ The  lower bound is embedded into the branch-and-bound scheme presented by \cite{arigliano2018branch} that we sketch here.
First we initialize the subproblem queue $Q$ with the original problem $P$ and set the upper bound  $UB=z(\underline{p}^*_{\textbf{y}^*},t_0)$. The instances of the Asymmetric TSP are solved by means of \cite{carpaneto1980some}.
At a generic step, let $\tilde{P}$ be the subproblem extracted from $Q$ and characterized by a set of branching constraints corresponding to a \textit{fixed} path $p_k=(0=v_0, v_1, v_2,\dots$ $v_{k})$, and possibly some forbidden arcs. Then we compute its lower bound, named $LB_1$. In particular we run the proposed procedure by omitting the first initialization step: the step function $\textbf{y}^*$ and the functions $x^*_{ij}(t)$ are inherited by the parent problem, whilst the coefficient $\underline{c}_{ij}$ is updated as follows:  
$$\underline{c}_{ij}=\min\limits_{t\in[z(p_k,t_0),UB]}x^*_{ij}(t),$$
with $(i,j)\in A$. Hence, we examine the (integer) optimal solution of the $\underline{\tilde{P}}$ relaxation:
\begin{equation}\label {LB1sol}
(0=v_0, v_1, \dots, v_{k}, v_{k+1}, \dots, v_{k^{\prime}}, v_{k^{\prime}+1}, \dots, v_{n+1}=0), 
\end{equation} 
which also gives an upper bound $UB_1$ on the optimal solution of $\tilde{P}$, obtained by evaluating (\ref{LB1sol}) w.r.t. the original travel time function $\tau$. If  $UB_1<UB$,  then we determine the upper bound as $UB=UB_1$. If $UB_1=LB_1$, then this solution is optimal for problem $\tilde{P}$. Otherwise, if $LB_1<UB$, the procedure branches as follows. Let $k^{\prime} \in \{k+1,\dots,n\}$ be the index such that its traversal time evaluated w.r.t. $\underline{\tau}$ is equal to the original traversal time, i.e. $\underline{z}(p_{k'},t_0)=z(p_{k'},t_0)$. 
We create $k'-k+1$ subproblems as follows:
\begin{equation}
\psi_{v_{k},v_{k+1}}=1, \psi_{v_{k+1},v_{k+2}}=1, \psi_{v_{k+2},v_{k+3}}=1, \dots, \psi_{v_{k'},v_{k'+1}}=1;
\end{equation}
\begin{equation} \nonumber
\dots
\end{equation}
\begin{equation}
\psi_{v_{k},v_{k+1}}=1, \psi_{v_{k+1},v_{k+2}}=1, \psi_{v_{k+2},v_{k+3}}=0; 
\end{equation}
\begin{equation}
\psi_{v_{k},v_{k+1}}=1, \psi_{v_{k+1},v_{k+2}}=0;
\end{equation}
\begin{equation}
\psi_{v_{k},v_{k+1}}=0,
\end{equation}
where $\psi_{ij}$ is a binary variable equal to 1 if and only if arc $(i,j) \in A$ is in the solution, while the arcs having $\psi_{ij}$ equal to 0 are forbidden. \\

\indent  We compared the obtained results with the results obtained in \cite{ADAMO2020104795}. Both procedures are implemented in C++ and run on MacBook Pro with a 2.33-GHz Intel Core 2 Duo processor and 4 GB of memory. Linear programs are solved with Cplex 12.9. 
We consider the same instances generated in \cite{ADAMO2020104795}  and we impose a time limit of 3600 s. Two scenarios are given: a first traffic pattern A in which a limited traffic zone is located in the center; a second traffic pattern B in which a heaviest traffic congestion is situated in the center. In  \cite{ADAMO2020104795}  a fundamental role was played by $\Delta$ which represents the worst degradation $\delta_{ijh}$ of the congestion factor of any arc $(i, j) \in A$ over the entire planning horizon, i.e. $\Delta=\min\limits_{h=0,\dots,H-1}(\min\limits_{(i,j)\in A}\delta_{ijh})$. For \cite{ADAMO2020104795} hard instances are characterized by low value of $\Delta$ in traffic pattern B. The results for the two scenarios are shown in Tables \ref{comp-resA} and \ref{comp-resB} in which 30 instances are generated for each combination of $|V | = 15, 20, 25, 30, 35, 40, 45, 50$ and $\Delta= 0.90, 0.80, 0.70$. The headings are as follows:
\begin{itemize}
\item $OPT$: number of instances solved to optimality out of 30;
\item $UB_I/LB_F$: average ratio of the initial upper bound value $UB_I$ the best lower bound $LB_F$ available at the end of the search;
\item $GAP_I$ : average initial optimality gap $\frac{UB_I-LB_I}{LB_I}\%$;
\item $GAP_F$ : average final optimality gap $\frac{UB_F-LB_F}{LB_F}\%$;
\item $NODES$: average number of nodes;
\item $TIME$: average computing time in seconds.
\end{itemize}
Except for columns $OPT$, we report results on two distinct rows: the first row is the average across instances solved to optimality, and the second row is the average for the remaining instances. For the sake of conciseness, the first or the second row has been omitted whenever none or all instances are solved to optimality. For columns  $NODES$ and $TIME$ we report only averages for instances that are solved to optimality.
Computational results show that our algorithm is capable of solving \textcolor{black}{923} instances out of \textcolor{black}{1440} instances while the procedure by \cite{ADAMO2020104795} solved only \textcolor{black}{566} problems. The improvement is remarkable for hard instances characterized by the traffic pattern B. Here, the \cite{ADAMO2020104795}  algorithm solves \textcolor{black}{229} out of \textcolor{black}{720} instances while our algorithm procedure succeeds on \textcolor{black}{505} instances. This can be explained by the quality of the new lower bound as well as by the low computational effort spent to solve the linear program (\ref{obj})-(\ref{c7}).  In particular,  we formulate a reduced-size instance of such linear model by including in the set $\mathcal{B}$ only 75 time instants. \textcolor{black}{Moreover we set the value of $\rho$ equal to $1/\min\limits_{h=0,\dots,|\mathcal{B}|-1}(t_{h+1}-t_h)$}. We finally observe that  the average computational time  is \textcolor{black}{438} seconds  for our procedure and \textcolor{black}{710} seconds for \cite{ADAMO2020104795}. Indeed, in spite of a higher initial gap, the new lower bound reduces the overall number of visited branch-and-bound nodes:  on average our procedure processes \textcolor{black}{413} nodes (i.e. ATSP instances) less than  \cite{ADAMO2020104795}.
\begin{landscape}

\begin{table}[H]
\tiny
\caption{Pattern A}
\label{comp-resA}
\begin{tabular}{|c|c|c|c|c|c|c|c|c|c|c|c|c|c|}
\cline{1-14}
\multicolumn{2}{|c|}{}         &  \multicolumn{6}{c|}{\cite{arigliano2018branch} branch-and-bound with   \cite{ADAMO2020104795} LB }                                                                       & \multicolumn{6}{c|}{\cite{arigliano2018branch} branch-and-bound with the new LB}                                                                                                          \\ \cline{1-14}
$\Delta$ & $|V|$               & $OPT$               & $UB_I/LB_F$ & $GAP_{I}$ & $GAP_{F}$ & $NODES$ & $TIME$  & $OPT$ & $UB_I/LB_F$ & $GAP_{I}$ & $GAP_{F}$ & $NODES$ & $TIME$  \\ \cline{1-14}
\multirow{16}{*}{0.70}      & \multirow{2}{*}{15} & \multirow{2}{*}{28}   & 1.014       & 7.218     & 0.000     & 2137    & 296.89  & \multirow{2}{*}{30}    &                                                               1.014       & 7.034     & 0.000     & 117   & 29.98   \\
&                     &                       & 1.002       & 7.523     & 7.343     & 42088   & --      &       &      --       &       --       &  --         &     --    &  --       \\ \cline{2-14}
         & \multirow{2}{*}{20} & \multirow{2}{*}{21}   & 1.007       & 4.998     & 0.000     & 3784    & 840.85  & \multirow{2}{*}{27}    &                                                               1.009       & 5.662     & 0.000     & 494   & 145.50  \\
         &                     &                       & 1.010       & 9.782     & 7.425     & 20209   & --      &     &                                                               1.002       & 7.605     & 7.542     & 13574 & -- \\ \cline{2-14}
         & \multirow{2}{*}{25} & \multirow{2}{*}{10}   & 1.012       & 4.083     & 0.000     & 4494    & 1537.26 & \multirow{2}{*}{21}    &                                                               1.010       & 5.141     & 0.000     & 2387  & 872.45  \\
         &                     &                       & 1.009       & 6.412     & 5.135     & 14500   & --      &      &                                                               1.005       & 6.877     & 6.490     & 10339 & -- \\ \cline{2-14}
         & \multirow{2}{*}{30} & \multirow{2}{*}{7}    & 1.011       & 2.821     & 0.000     & 3648    & 1865.11 & \multirow{2}{*}{10}    &                                                              1.010       & 4.696     & 0.000     & 2067  & 1522.54 \\
         &                     &                       & 1.005       & 5.812     & 5.129     & 14979   & --      &      &                                                                1.008       & 5.379     & 4.625     & 6081  & -- \\ \cline{2-14}
         & \multirow{2}{*}{35} & \multirow{2}{*}{3}    & 1.004       & 1.897     & 0.000     & 672     & 1120.41 & \multirow{2}{*}{2}     &                                                              1.009       & 3.371     & 0.000     & 330   & 1451.72 \\
         &                     &                       & 1.007       & 4.410     & 3.749     & 11223   & --      &      &                                                               1.006       & 5.511     & 4.989     & 2794   & -- \\ \cline{2-14}
         & \multirow{2}{*}{40} & \multirow{2}{*}{3}    & 1.013       & 2.494     & 0.000     & 3460    & 659.11  &   \multirow{2}{*}{3}    &                                                                       1.008      &      5.319     &    0.000       &       307  &     1701.20    \\
         &                     &                       & 1.007       & 5.468     & 4.746     & 14572   & --      &      &                                                                1.003       & 5.969     & 5.712     & 4981   & -- \\ \cline{2-14}
         & \multirow{2}{*}{45} & \multirow{2}{*}{7}    & 1.014       & 1.984     & 0.000     & 1413    & 977.83  &   \multirow{2}{*}{7}    &                                                                         1.006    &     1.379     &     0      &    321     &  992.62       \\
         &                     &                       & 1.006       & 4.942     & 4.347     & 12117   & --      &     &                                                                1.002       & 5.646    & 5.505     & 6568   & -- \\ \cline{2-14}
         & \multirow{2}{*}{50} & \multirow{2}{*}{3}    & 1.001       & 0.686     & 0.000     & 64      & 214.06  &  \multirow{2}{*}{3}     &                                                                       1.004      &    2.906       &    0.000       &  153       &  604.27       \\
         &                     &                       & 1.003       & 5.842     & 5.576     & 22682   & --      &       &                                                                      1.003      &    7.092       & 6.851          &    11408     &       --  \\ \cline{1-14}
\multirow{16}{*}{0.80}      & \multirow{2}{*}{15} & \multirow{2}{*}{30}    & 1.007       & 4.528     & 0.000     & 1829    & 189.04  &  \multirow{2}{*}{30}     &                                                                          1.007   &   4.259        &     0.000      &   172      &      4.54   \\
         &                     &                       & --          & --        & --        & --      & --      &       &                                                                       --      & --          &     --      &   --      &   --      \\ \cline{2-14}
         & \multirow{2}{*}{20} & \multirow{2}{*}{27}   & 1.003       & 3.500     & 0.000     & 1654    & 559.23  &  \multirow{2}{*}{30}     &                                                                        1.004     &   3.491        &   0.000        &      608   &      163.98   \\
         &                     &                       & 1.002       & 5.962     & 5.768     & 23021   & --      &       &                                                                        --     &    --       &   --        &    --     &   --      \\ \cline{2-14}
         & \multirow{2}{*}{25} & \multirow{2}{*}{14}   & 1.006       & 3.282     & 0.000     & 3036    & 1072.43 &   \multirow{2}{*}{26}    &                                                                        1.006     &    3.431       &    0.000       &   997     &  367.12       \\
         &                     &                       & 1.005       & 3.813     & 3.543     & 11226   & --      &       &                                                                         1.003    &      4.639     &      4.432     &    10523     &      --   \\ \cline{2-14}
         & \multirow{2}{*}{30} & \multirow{2}{*}{9}    & 1.004       & 2.024     & 0.000     & 1689    & 1584.40 &   \multirow{2}{*}{21}    &                                                                      1.006       &    2.769       &     0.000      &    1439     &     846.80    \\
         &                     &                       & 1.002       & 3.903     & 3.730     & 12568   & --      &       &                                                                       1.003      &     3.833      &  3.666         &    6017     &        -- \\ \cline{2-14}
         & \multirow{2}{*}{35} & \multirow{2}{*}{3}    & 1.002       & 1.497     & 0.000     & 654     & 978.10  &  \multirow{2}{*}{9}     &                                                                           1.006  &     3.181      &    0.000       &  416       &   1682.50      \\
         &                     &                       & 1.005       & 2.868     & 2.542     & 6610    & --      &       &                                                                       1.003      &    3.288       &   2.962        &    7591     &      --   \\ \cline{2-14}
         & \multirow{2}{*}{40} & \multirow{2}{*}{7}    & 1.005       & 1.809     & 0.000     & 2822    & 1188.50 &    \multirow{2}{*}{8}   &                                                                       1.001      &    3.028       &   0.000        & 318        &        1752.77 \\
         &                     &                       & 1.005       & 3.981     & 3.542     & 11260   & --      &       &                                                                    1.002         &   3.700        & 3.530          &   8958      &      --   \\ \cline{2-14}
         & \multirow{2}{*}{45} & \multirow{2}{*}{7}    & 1.006       & 1.397     & 0.000     & 631     & 990.52  &   \multirow{2}{*}{8}    &                                                                      1.003       &    3.182       &   0.000        &  403       &        515.17 \\
         &                     &                       & 1.003       & 3.435     & 3.157     & 9838    & --      &       &                                                                    1.002         &     3.493      &    3.260       &    6047     &        -- \\ \cline{2-14}
         & \multirow{2}{*}{50} & \multirow{2}{*}{4}    & 1.002       & 0.830     & 0.000     & 380     & 945.71  &  \multirow{2}{*}{4}     &                                                                      1.004      &   3.087        &   0.000        &   226      &         576.84 \\
         &                     &                       & 1.001       & 3.976     & 3.873     & 16293   & --      &       &                                                                      1.001       &   4.311        &   4.184        &  914       &         -- \\ \cline{1-14}
\multirow{14}{*}{0.90}      & 15                     &  30      & 1.002       & 2.218     & 0.000     & 282     & 29.86   &  30     &                                                                      1.003       &     1.842      &   0.000        &    24     & 0.88        \\ \cline{2-14}
         & 20                  & 30                    & 1.001       & 1.883     & 0.000     & 642     & 164.83  &  30     &                                                                          1.003   &      1.603     &   0.000        &  205       &      6.20   \\ \cline{2-14}
         & \multirow{2}{*}{25} & \multirow{2}{*}{27}   & 1.002       & 1.833     & 0.000     & 1617    & 771.20  & \multirow{2}{*}{30}     &                                                                            1.003 &       1.662    & 0.000          &  1329       &       58.79\\
         &                     &                       & 1.004       & 2.671     & 2.295     & 7931    & --      &         &                                                                           --  &     --      &     --      &   --      &  --        \\ \cline{2-14}
         & \multirow{2}{*}{30} & \multirow{2}{*}{23}   & 1.002       & 1.563     & 0.000     & 2109    & 1193.32 &  \multirow{2}{*}{30}     &                                                                           1.002  &  1.323         &  0.000         &  500       &     385.08    \\
         &                     &                       & 1.002       & 1.964     & 1.611     & 4466    & --      &       &                                                                          --   &      --     &    --       &   --      &   --      \\ \cline{2-14}
         & \multirow{2}{*}{35} & \multirow{2}{*}{16}   & 1.001       & 1.272     & 0.000     & 2238    & 1913.16 &   \multirow{2}{*}{23}     &                                                                       1.003      &   1.346        &     0.000      &   1691      &     843.76    \\
         &                     &                       & 1.004       & 1.579     & 1.302     & 7903    & --      &       &                                                                         1.000    &   1.609        &  1.598         &   6580      &      --   \\ \cline{2-14}
         & \multirow{2}{*}{40} & \multirow{2}{*}{11}   & 1.002       & 1.201     & 0.000     & 1999    & 1396.01 &   \multirow{2}{*}{13}     &                                                                        1.002     &     1.300      & 0.000          & 1941        &     765.65    \\
         &                     &                       & 1.001       & 1.955     & 1.766     & 8641    & --      &       &                                                                        1.001     &      1.734     &    1.621       &  6429       &    --     \\ \cline{2-14}
         & \multirow{2}{*}{45} & \multirow{2}{*}{11}   & 1.002       & 1.043     & 0.000     & 808     & 971.56  &    \multirow{2}{*}{14}    &                                                                          1.002   &      1.476     &   0.000        & 2800        &      864.54   \\
         &                     &                       & 1.001       & 1.732     & 1.570     & 8021    & --      &       &                                                                      1.002      &    1.717       &    1.499       & 5247        &      --   \\ \cline{2-14}
         & \multirow{2}{*}{50} & \multirow{2}{*}{6}    & 1.003       & 0.927     & 0.000     & 1377    & 733.67  &    \multirow{2}{*}{9}    &                                                                       1.002      &    1.675       &   0.000        &  2980       &  1623.40       \\
         &                     &                       & 1.000       & 2.029     & 1.877     & 7937    & --      &       &                                                                      1.001       &      2.267     &  2.109         &   7594      &       --  \\ \cline{1-14}
AVG      &                     & 337                   & 1.005       & 2.910     & 0.000     & 1816    & 754.91  &   418     &             1.005                                                                &        3.127   &    0.000       &   894      &   461.05      \\ \cline{1-14}
\end{tabular}
\end{table}
\end{landscape}

\begin{landscape}
\begin{table}[H]
\tiny
\caption{Pattern B}
\label{comp-resB}
\begin{tabular}{|c|c|c|c|c|c|c|c|c|c|c|c|c|c|}
\hline
\multicolumn{2}{|c|}{}                       & \multicolumn{6}{c|}{\cite{arigliano2018branch} branch-and-bound with   \cite{ADAMO2020104795} LB}                                                                         & \multicolumn{6}{c|}{\cite{arigliano2018branch} branch-and-bound with the new LB}                                                                                                          \\ \hline
$\Delta$               & $|V|$               & $OPT$               & $UB_I/LB_F$ & $GAP_{I}$ & $GAP_{F}$ & $NODES$ & $TIME$ & $OPT$ & $UB_I/LB_F$ & $GAP_{I}$ & $GAP_{F}$ & $NODES$ & $TIME$  \\ \hline
\multirow{20}{*}{0.70} & \multirow{2}{*}{15} & \multirow{2}{*}{28} & 1.011       & 5.550     & 0.000     & 5074    & 643.15 &\multirow{2}{*}{30}&1.010&9.158&0.000&441&11.72 \\
                       &                     &                     & 1.000       & 21.579    & 9.327     & 66059   & --     &&--&--&--&--&-- \\ \cline{2-14}
                       & \multirow{2}{*}{20} & \multirow{2}{*}{14} & 1.007       & 4.508     & 0.000     & 6362    & 1348.98&\multirow{2}{*}{30}&1.008&6.844&0.000&1199&43.03 \\
                       &                     &                     & 1.004       & 15.856    & 5.468     & 51346   & --     &&--&--&--&--&-- \\ \cline{2-14}
                       & \multirow{2}{*}{25} & \multirow{2}{*}{4}  & 1.016       & 3.243     & 0.000     & 847     & 526.58 &\multirow{2}{*}{29}&1.011&5.122&0.000&5523&456.17 \\
                       &                     &                     & 1.007       & 11.152    & 4.581     & 35810   & --     &&1.000&6.114&5.989&49813&-- \\ \cline{2-14}
                       & \multirow{2}{*}{30} & \multirow{2}{*}{2}  & 1.003       & 0.317     & 0.000     & 48      & 14.03  &\multirow{2}{*}{23}&1.012&4.673&0.000&3582&715.48 \\
                       &                     &                     & 1.007       & 9.454     & 4.688     & 41335   & --     &&1.014&6.081&4.705&24953&-- \\ \cline{2-14}
                       & \multirow{2}{*}{35} & \multirow{2}{*}{2}  & 1.000       & 0.000     & 0.000     & 7       & 14.18  &\multirow{2}{*}{10}&1.016&5.593&0.000&27587&2037.09 \\
                       &                     &                     & 1.008       & 7.583     & 5.083     & 39228   & --     &&1.011&4.899&3.874&33274&-- \\ \cline{2-14}
                       & \multirow{2}{*}{40} & \multirow{2}{*}{0}  & --          & --        & --        & --      & --     &\multirow{2}{*}{3}&1.003&3.684&0.000&31660&2182.29 \\
                       &                     &                     & 1.002       & 7.910     & 4.856     & 28032   & --     &&1.005&5.139&4.771&55949&-- \\ \cline{2-14}
                       & \multirow{2}{*}{45} & \multirow{2}{*}{1}  & 1.003       & 0.343     & 0.000     & 283     & 380.00 &\multirow{2}{*}{2}&1.012&4.083&0.000&11000&2439.76 \\
                       &                     &                     & 1.002       & 8.806     & 5.793     & 15943   & --     &&1.005&6.060&5.811&65674&-- \\ \cline{2-14}
                       & \multirow{2}{*}{50} & \multirow{2}{*}{2}  & 1.000       & 0.000     & 0.000     & 0       & 5.16   &\multirow{2}{*}{2}&1.013&2.892&0.000&1532&911.54 \\
                       &                     &                     & 1.002       & 7.581     & 5.492     & 14310   & --     &&1.006&6.707&6.274&8871&-- \\ \cline{2-14}
 \hline\hline
\multirow{20}{*}{0.80} & \multirow{2}{*}{15} & \multirow{2}{*}{30} & 1.007       & 4.145     & 0.000     & 4172    & 287.60       &\multirow{2}{*}{30}&1.009&6.435&0.000&167&4.59\\
                       &                     &                     &         --    &     --      &      --     &    --     &  --  &&--&--&--&--&--     \\ \cline{2-14}
                       & \multirow{2}{*}{20} & \multirow{2}{*}{19} & 1.008       & 3.950     & 0.000     & 6232    & 1033.32      &\multirow{2}{*}{30}&1.007&4.936&0.000&1016&33.55\\
                       &                     &                     & 1.001       & 3.836     & 3.755     & 55019   & --           &&--&--&--&--&-- \\ \cline{2-14}
                       & \multirow{2}{*}{25} & \multirow{2}{*}{7}  & 1.007       & 2.446     & 0.000     & 6253    & 1265.33      &\multirow{2}{*}{30}&1.006&3.498&0.000&4619&268.10\\
                       &                     &                     & 1.004       & 5.920     & 3.212     & 28420   & --           &&--&--&--&--&--\\ \cline{2-14}
                       & \multirow{2}{*}{30} & \multirow{2}{*}{4}  & 1.001       & 1.264     & 0.000     & 867     & 763.58       &\multirow{2}{*}{26}&1.008&3.141&0.000&1438&374.71\\
                       &                     &                     & 1.003       & 4.905     & 3.202     & 20060   & --           &&1.007&3.603&2.889&35382&--\\ \cline{2-14}
                       & \multirow{2}{*}{35} & \multirow{2}{*}{2}  & 1.000       & 0.000     & 0.000     & 10      & 13.43        &\multirow{2}{*}{21}&1.008&2.987&0.000&1012&997.28 \\
                       &                     &                     & 1.002       & 4.517     & 3.696     & 21260   & --           &&1.004&3.290&2.964&23654&-- \\ \cline{2-14}
                       & \multirow{2}{*}{40} & \multirow{2}{*}{0}  & --          & --        & --        & --      & --           &\multirow{2}{*}{13}&1.007&2.704&0.000&1823&1278.25 \\
                       &                     &                     & 1.002       & 5.136     & 3.633     & 17467   &              &&1.003&3.071&2.800&41341&--  \\ \cline{2-14}
                       & \multirow{2}{*}{45} & \multirow{2}{*}{1}  & 1.003       & 0.339     & 0.000     & 186     & 187.29       &\multirow{2}{*}{9}&1.019&3.759&0.000&7764&1047.82 \\
                       &                     &                     & 1.001       & 5.740     & 4.074     & 14109   & --           &&1.005&3.475&3.085&26611&-- \\ \cline{2-14}
                       & \multirow{2}{*}{50} & \multirow{2}{*}{2}  & 1.000       & 0.000     & 0.000     & 0       & 5.59         &\multirow{2}{*}{7}&1.011&2.803&0.000&5906&1316.80 \\
                       &                     &                     & 1.001       & 4.792     & 3.742     & 9640    & --           &&1.006&4.003&3.450&44414&--  \\ \cline{2-14}
 \hline\hline
\multirow{19}{*}{0.90} & 15                  & 30                  & 1.003       & 1.898     & 0.000     & 358     & 29.58   &30&1.004&3.099&0.000&306&1.05 \\ \cline{2-14}
                       & \multirow{2}{*}{20} & \multirow{2}{*}{30} & 1.003       & 1.904     & 0.000     & 3271    & 351.53  &\multirow{2}{*}{30}&1.003&2.526&0.000&103&3.93\\
                       &                     &                     & --          & --        & --        & --      & --      &&--&--&--&--&--\\ \cline{2-14}
                       & \multirow{2}{*}{25} & \multirow{2}{*}{22} & 1.002       & 1.359     & 0.000     & 4406    & 1108.16 &\multirow{2}{*}{30}&1.003&1.870&0.000&400&28.19\\
                       &                     &                     & 1.002       & 2.698     & 1.857     & 21989   & --      &&--&--&--&--&--\\ \cline{2-14}
                       & \multirow{2}{*}{30} & \multirow{2}{*}{15} & 1.003       & 1.345     & 0.000     & 3668    & 1660.69 &\multirow{2}{*}{30}&1.003&1.587&0.000&793&174.58\\
                       &                     &                     & 1.001       & 2.077     & 1.531     & 22643   & --      &&--&--&--&--&--\\ \cline{2-14}
                       & \multirow{2}{*}{35} & \multirow{2}{*}{9}  & 1.002       & 0.878     & 0.000     & 1608    & 1440.52 &\multirow{2}{*}{28}&1.003&1.404&0.000&3684&402.04\\
                       &                     &                     & 1.000       & 2.091     & 1.633     & 15296   & --      &&1.001&1.802&1.801&16242&--\\ \cline{2-14}
                       & \multirow{2}{*}{40} & \multirow{2}{*}{1}  & 1.000       & 1.680     & 0.000     & 1768    & 2473.61 &\multirow{2}{*}{25}&1.002&1.276&0.000&6130&633.34\\
                       &                     &                     & 1.001       & 2.148     & 1.565     & 12088   & --      &&1.000&1.792&1.783&27452&--\\ \cline{2-14}
                       & \multirow{2}{*}{45} & \multirow{2}{*}{2}  & 1.001       & 0.088     & 0.000     & 45      & 162.17  &\multirow{2}{*}{20}&1.006&1.709&0.000&5914&994.04\\
                       &                     &                     & 1.001       & 2.357     & 1.810     & 9494    & --      &&1.003&1.639&1.408&20091&--\\ \cline{2-14}
                       & \multirow{2}{*}{50} & \multirow{2}{*}{2}  & 1.000       & 0.000     & 0.000     & 0       & 4.99    &\multirow{2}{*}{17}&1.004&1.753&0.000&8437&1072.50\\
                       &                     &                     & 1.001       & 2.138     & 1.706     & 6766    & --      &&1.002&1.774&1.577&30276&--\\ \cline{2-14}
 \hline\hline
AVG                    &                     & 229                 & 1.005       & 2.792     & 0.000     & 3561    & 665.36  &505 & 1.007& 3.770 & 0.000 & 3210 & 418.86\\
 \hline
\end{tabular}
\end{table}
\end{landscape}

\subsection {Results on the TDRPP}
 \cite{calogiuri2019branch} devised lower and upper bounds for the TDRPP. The proposed bounds were embedded in a branch-and-bound algorithm.  In particular, they enhanced the branching scheme of the  algorithm by \cite{arigliano2018branch}  in order to deal with quickest paths linking distinct required arcs. \\
 We embedded our bounds in such optimal solution algorithm. Both procedures are implemented in C++ and run on MacBook Pro-with a 2.33-GHz Intel Core 2 Duo processor and 4 GB of memory. Linear programs are solved with Cplex 12.9. In particular the algorithm by  \cite{avila2015new}  was used to solve the time-invariant routing problem at the second step of the proposed lower bounding procedure.\\
We consider the  instances that \cite{calogiuri2019branch} generated in two phases. First, a complete graph $G = (V, A, \tau )$ was generated where the $|V|$ locations are randomly generated as reported in  \cite{Cordeau2014}. For each possible value of $|V|$ in the set $\{20, 40, 60\}$. The set of arcs A was determined by selecting $|V | + \alpha \times |V |$ arcs from $A$ as follows. The set  A was first initialized with the $|V|$ arcs of an Hamiltonian tour on $G$ . Then  a number of arcs $\alpha \times |V|$ are added, where $\alpha\in \{0.2; 0.6; 1\}$. The required arcs set $R$ was generated by randomly selecting a percentage $\beta$ of arcs in $A$, with $\beta \in \{0.3; 0.5; 0.7\}$. 
By combining 30 instances for each value of $|V|$, three values of $|V|$, three values of $\alpha$ and three values of $\beta$, \cite{calogiuri2019branch} obtained 810 different combinations of $(G, R)$. For each combination $(G, R)$, they generated 3 different traffic patterns, characterized by three different values of traffic congestion, i.e. $\Delta=\{0.7,0.8,0.9\}.$ 
Results are reported in Tables \ref{TDRPP_alpha03},\ref{TDRPP_alpha05} and \ref{TDRPP_alpha07} with the same column headings of TDTSP results.
Computational results show that our algorithm is capable of solving \textcolor{black}{1772} instances out of \textcolor{black}{2430} instances while the procedure by \cite{calogiuri2019branch} solved only \textcolor{black}{1618} problems. The improvement is remarkable for computational time, which was on average equal to  \textcolor{black}{98} seconds  for our procedure and \textcolor{black}{189} seconds for  \cite{calogiuri2019branch}. This can be explained by the quality of the new lower bound as well as by the low computational effort spent to solve the linear program (\ref{obj})-(\ref{c7}), which was formulated by including in the set $\mathcal{B}$ only 75 time instants. 
Indeed, the new lower bound reduces the overall number of visited branch-and-bound nodes:  on average our procedure processes about \textcolor{black}{4394} nodes (i.e. RPP instances) less than  \cite{calogiuri2019branch}.

\section {Conclusion}\label{sec_6}
This paper has introduced a key property of time-dependent graphs that we have called "path ranking invariance''. We have shown that this property can be exploited in order to solve a \textcolor{black}{large} class of time-dependent routing problems, including the \textit{Time-Dependent Travelling Salesman Problem} and the \textit{Time-Dependent Rural Postman Problem}. We have also proved that the path ranking invariance can be checked  by solving a decision problem, named \textit{Constant Traversal Cost Problem}. When such check fails,  we have defined a new family of parameterized lower and upper bounds, whose parameters can be  chosen by solving a linear program. Computational results show that, used in a branch-and-bound algorithm, this mechanism outperforms state-of-the-arts optimal algorithms. Future work will focus on adapting the new ideas to other time-dependent routing problems.
\begin{sidewaystable}[h]
\renewcommand{\tabcolsep}{3pt}
\caption{Pattern B. $\beta = 0.3$} \label{TDRPP_alpha03}
\begin{center}
\scriptsize
\begin{tabular}{|c|c|c||c|c|c|c|c|c||c|c|c|c|c|c|}
\hline
\multicolumn{3}{|c||}{}                       & \multicolumn{6}{|c||}{ \cite{calogiuri2019branch} branch-and-bound} & \multicolumn{6}{|c|}{ \cite{calogiuri2019branch} branch-and-bound with the new LB}                                                                                                          \\ \hline

\cline{1-15}
Delta                 & $|V|$               & $|A|$                & $OPT$          & $UB_I/LB_F$ & $GAP_I$ & $GAP_F$ & $NODES$ & $TIME$ & $OPT$ & $UB_I/LB_F$ & $GAP_I$ & $GAP_F$ & $NODES$ & $TIME$ \\ \cline{1-15}
\multirow{14}{*}{0.9} & \multirow{4}{*}{20} & 24                   & 30                  & 1.001       & 4.37         & 0.00         & 11           & 0.40        & 30         & 1.000       & 0.52       & 0.00       & 9            & 0.38        \\ \cline{3-15}
                      &                     & 32                   & 30                  & 1.000       & 4.83         & 0.00         & 1173         & 30.57       & 30         & 1.000       & 0.48       & 0.00       & 80           & 3.17        \\ \cline{3-15}
                      &                     & \multirow{2}{*}{40}  & \multirow{2}{*}{29} & 1.000       & 4.26         & 0.00         & 835          & 22.81       & \multirow{2}{*}{29}         & 1.000       & 0.45       & 0.00       & 763          & 21.79       \\
                      &                     &                      &                     & 1.000       & 5.19         & 0.30         & -            & -           &           & 1.000       & 0.14       & 0.14       & -            & -           \\ \cline{2-15}
                      & \multirow{5}{*}{40} & 48                   & 30                  & 1.000       & 4.39         & 0.00         & 127          & 3.80        & 30         & 1.001       & 0.46       & 0.00       & 65           & 2.18        \\ \cline{3-15}
                      &                     & \multirow{2}{*}{64}  & \multirow{2}{*}{27} & 1.000       & 4.49         & 0.00         & 5212         & 148.00      & \multirow{2}{*}{27}         & 1.000       & 0.45       & 0.00       & 2294         & 83.66       \\
                      &                     &                      &                     & 1.000       & 5.47         & 0.57         & -            & -           &           & 1.000       & 1.06       & 1.04       & -            & -           \\ \cline{3-15}
                      &                     & \multirow{2}{*}{80}  & \multirow{2}{*}{26} & 0.01        & 4.25         & 0.00         & 8415         & 216.75      & \multirow{2}{*}{27}         & 1.000       & 0.48       & 0.00       & 0            & 0.11        \\
                      &                     &                      &                     & 1.000       & 5.42         & 0.54         & -            & -           &           & 1.000       & 0.25       & 0.24       & -            & -           \\ \cline{2-15}
                      & \multirow{5}{*}{60} & 72                   & 30                  & 1.001       & 4.73         & 0.00         & 1029         & 29.05       & 30         & 1.000       & 0.38       & 0.00       & 993          & 31.80       \\ \cline{3-15}
                      &                     & \multirow{2}{*}{96}  & \multirow{2}{*}{24} & 1.000       & 4.61         & 0.00         & 1927         & 57.87       & \multirow{2}{*}{25}         & 1.000       & 0.09       & 0.00       & 2            & 0.18        \\
                      &                     &                      &                     & 1.000       & 5.78         & 0.51         & -            & -           &           & 1.000       & 0.65       & 0.65       & -            & -           \\ \cline{3-15}
                      &                     & \multirow{2}{*}{120} & \multirow{2}{*}{28} & 1.000       & 4.81         & 0.00         & 2513         & 82.47       & \multirow{2}{*}{30}         & \multirow{2}{*}{1.000}       & \multirow{2}{*}{0.27}       & \multirow{2}{*}{0.00}       & \multirow{2}{*}{6}            & \multirow{2}{*}{0.47}        \\
                      &                     &                      &                     & 1.000       & 5.57         & 0.43         & -            & -           &            &             &              &              &              &             \\ \cline{1-15}
\multirow{13}{*}{0.8} & \multirow{3}{*}{20} & 24                   & 30                  & 1.000       & 7.94         & 0.00         & 10           & 0.38        & 30         & 1.000       & 1.22       & 0.00       & 10           & 0.37        \\ \cline{3-15}
                      &                     & 32                   & 30                  & 1.000       & 8.07         & 0.00         & 151          & 3.79        & 30         & 1.001       & 1.16       & 0.00       & 79           & 3.40        \\ \cline{3-15}
                      &                     & 40                   & 30                  & 1.000       & 7.58         & 0.00         & 3626         & 86.90       & 30         & 1.000       & 1.26       & 0.00       & 875          & 34.78       \\ \cline{2-15}
                      & \multirow{5}{*}{40} & 48                   & 30                  & 1.000       & 8.29         & 0.00         & 71           & 1.95        & 30         & 1.000       & 1.22       & 0.00       & 46           & 2.01        \\ \cline{3-15}
                      &                     & \multirow{2}{*}{64}  & \multirow{2}{*}{27} & 1.001       & 7.63         & 0.00         & 9451         & 269.27      & \multirow{2}{*}{27}         & 1.000       & 0.76       & 0.00       & 5124         & 211.62      \\
                      &                     &                      &                     & 1.002       & 8.86         & 1.75         & -            & -           &           & 1.000       & 1.88       & 1.71       & -            & -           \\ \cline{3-15}
                      &                     & \multirow{2}{*}{80}  & \multirow{2}{*}{26} & 1.000       & 7.73         & 0.00         & 8618         & 183.08      & \multirow{2}{*}{27}         & 1.000       & 0.69       & 0.00       & 7            & 0.35        \\
                      &                     &                      &                     & 1.000       & 8.38         & 0.84         & -            & -           &           & 1.000       & 0.90       & 0.80       & -            & -           \\ \cline{2-15}
                      & \multirow{5}{*}{60} & 72                   & 30                  & 1.000       & 8.84         & 0.00         & 260          & 8.21        & 30         & 1.001       & 0.90       & 0.00       & 162          & 5.58        \\ \cline{3-15}
                      &                     & \multirow{2}{*}{96}  & \multirow{2}{*}{22} & 1.000       & 7.81         & 0.00         & 5372         & 146.27      & \multirow{2}{*}{25}         & 1.000       & 0.66       & 0.00       & 6            & 0.48        \\
                      &                     &                      &                     & 1.000       & 8.90         & 1.33         & -            & -           &           & 1.000       & 1.02       & 1.00       & -            & -           \\ \cline{3-15}
                      &                     & \multirow{2}{*}{120} & \multirow{2}{*}{25} & 1.000       & 7.06         & 0.00         & 4872         & 124.92      & \multirow{2}{*}{27}         & 1.000       & 1.01       & 0.00       & 2            & 0.32        \\
                      &                     &                      &                     & 1.000       & 9.29         & 0.86         & -            & -           &           & 1.000       & 2.01       & 2.00       & -            & -           \\ \cline{1-15}
\multirow{14}{*}{0.7} & \multirow{4}{*}{20} & 24                   & 30                  & 1.001       & 10.51        & 0.00         & 9            & 0.38        & 30         & 1.000       & 2.25       & 0.00       & 7           & 0.36        \\ \cline{3-15}
                      &                     & 32                   & 30                  & 1.000       & 10.88        & 0.00         & 385          & 11.69       & 30         & 1.001       & 2.09       & 0.00       & 318          & 11.03       \\ \cline{3-15}
                      &                     & \multirow{2}{*}{40}  & \multirow{2}{*}{29} & 1.000       & 12.48        & 0.00         & 4918         & 145.27      & \multirow{2}{*}{30}         & \multirow{2}{*}{1.000}       & \multirow{2}{*}{1.01}       & \multirow{2}{*}{0.00}       & \multirow{2}{*}{3332}         & \multirow{2}{*}{127.09}      \\
                      &                     &                      &                     & 1.000       & 14.89        & 2.47         & -            & -           &            &             &              &              &              &             \\ \cline{2-15}
                      & \multirow{5}{*}{40} & 48                   & 30                  & 1.002       & 11.02        & 0.00         & 37           & 1.37        & 30         & 1.001       & 1.64       & 0.00       & 31           & 1.23        \\ \cline{3-15}
                      &                     & \multirow{2}{*}{64}  & \multirow{2}{*}{25} & 1.000       & 12.62        & 0.00         & 13685        & 441.98      & \multirow{2}{*}{27}         & 1.000       & 2.55       & 0.00       & 4816         & 189.47      \\
                      &                     &                      &                     & 1.001       & 15.86        & 2.96         & -            & -           &           & 1.002       & 1.20       & 0.96       & -            & -           \\ \cline{3-15}
                      &                     & \multirow{2}{*}{80}  & \multirow{2}{*}{24} & 1.000       & 11.27        & 0.00         & 21551        & 481.22      & \multirow{2}{*}{26}         & 1.000       & 1.79       & 0.00       & 1818         & 74.14       \\
                      &                     &                      &                     & 1.001       & 13.85        & 1.13         & -            & -           &           & 1.000       & 0.55       & 0.52       & -            & -           \\ \cline{2-15}
                      & \multirow{5}{*}{60} & 72                   & 30                  & 1.001       & 10.49        & 0.00         & 499          & 14.38       & 30         & 1.001       & 1.45       & 0.00       & 352          & 13.37       \\ \cline{3-15}
                      &                     & \multirow{2}{*}{96}  & \multirow{2}{*}{20} & 1.000       & 12.58        & 0.00         & 21673        & 720.30      & \multirow{2}{*}{26}         & 1.000       & 1.43       & 0.00       & 2676         & 120.68      \\
                      &                     &                      &                     & 1.003       & 14.36        & 1.70         & -            & -           &           & 1.001       & 0.88       & 0.79       & -            & -           \\ \cline{3-15}
                      &                     & \multirow{2}{*}{120} & \multirow{2}{*}{24} & 1.000       & 12.08        & 0.00         & 25235        & 600.53      & \multirow{2}{*}{30}         & \multirow{2}{*}{1.000}       & \multirow{2}{*}{1.17}       & \multirow{2}{*}{0.00}       & \multirow{2}{*}{198}          & \multirow{2}{*}{12.44}       \\
                      &                     &                      &                     & 1.000       & 13.52        & 1.04         & -            & -           &            &             &              &              &              &             \\ \cline{1-15}
\multicolumn{3}{|c||}{AVG}                 & 746                 & 1.000       & 8.11         & 1.23         & 4668         & 125.71      & 773        & 1.000       & 1.03       & 0.04       & 862          & 34.04       \\ \cline{1-15}
\end{tabular}
\end{center}
\end{sidewaystable}

\begin{sidewaystable}[h]
\renewcommand{\tabcolsep}{3pt}
\caption{Pattern B. $\beta = 0.5$} \label{TDRPP_alpha05}
\begin{center}
\scriptsize
\begin{tabular}{|c|c|c||c|c|c|c|c|c||c|c|c|c|c|c|}
\hline
\multicolumn{3}{|c||}{}                       & \multicolumn{6}{|c||}{ \cite{calogiuri2019branch} branch-and-bound }& \multicolumn{6}{|c|}{ \cite{calogiuri2019branch} branch-and-bound with the new LB}                                                                                                          \\ \hline

\cline{1-15}
Delta                 & $|V|$               & $|A|$                & $OPT$          & $UB_I/LB_F$ & $GAP_I$ & $GAP_F$ & $NODES$ & $TIME$ & $OPT$ & $UB_I/LB_F$ & $GAP_I$ & $GAP_F$ & $NODES$ & $TIME$ \\ \cline{1-15}
\multirow{14}{*}{0.9} & \multirow{4}{*}{20} & 24                   & 30                  & 1.000       & 3.43         & 0.00         & 34           & 0.98        & 30         & 0.98      & 0.00       & 1.002        & 28           & 0.87        \\ \cline{3-15}
                      &                     & 32                   & 30                  & 1.001       & 4.17         & 0.00         & 2721         & 67.72       & 30         & 0.75      & 0.00       & 1.001        & 1083         & 35.80       \\ \cline{3-15}
                      &                     & \multirow{2}{*}{40}  & \multirow{2}{*}{22} & 1.000       & 3.81         & 0.00         & 22743        & 614.20      & \multirow{2}{*}{22}         & 0.55      & 0.00       & 1.000        & 13675        & 504.69      \\
                      &                     &                      &                     & 1.001       & 4.21         & 0.78         & -            & -           &           & 0.97      & 0.76       & 1.002        & -            & -           \\ \cline{2-15}
                      & \multirow{5}{*}{40} & 48                   & 30                  & 1.001       & 3.54         & 0.00         & 402          & 10.34       & 30         & 0.76      & 0.00       & 1.001        & 265          & 7.80        \\ \cline{3-15}
                      &                     & \multirow{2}{*}{64}  & \multirow{2}{*}{12} & 1.001       & 2.97         & 0.00         & 16729        & 453.01      & \multirow{2}{*}{17}         & 0.43      & 0.00       & 1.046        & 4588         & 245.90      \\
                      &                     &                      &                     & 1.001       & 3.85         & 0.75         & -            & -           &          & 0.92      & 0.73       & 1.002        & -            & -           \\ \cline{3-15}
                      &                     & \multirow{2}{*}{80}  & \multirow{2}{*}{9}  & 1.000       & 2.89         & 0.00         & 10158        & 260.85      & \multirow{2}{*}{16}         & 0.26      & 0.00       & 1.000        & 423          & 63.74       \\
                      &                     &                      &                     & 1.001       & 3.68         & 0.62         & -            & -           &          & 0.82      & 0.73       & 1.001        & -            & -           \\ \cline{2-15}
                      & \multirow{5}{*}{60} & 72                   & 30                  & 1.001       & 3.69         & 0.00         & 4050         & 116.01      & 30         & 0.69      & 0.00       & 1.001        & 1774         & 68.78       \\ \cline{3-15}
                      &                     & \multirow{2}{*}{96}  & \multirow{2}{*}{9}  & 1.000       & 3.27         & 0.00         & 30815        & 810.22      & \multirow{2}{*}{17}         & 0.25      & 0.00       & 1.000        & 21           & 0.43        \\
                      &                     &                      &                     & 1.001       & 3.78         & 0.54         & -            & -           &          & 0.72      & 0.66       & 1.001        & -            & -           \\ \cline{3-15}
                      &                     & \multirow{2}{*}{120} & \multirow{2}{*}{6}  & 1.000       & 3.23         & 0.00         & 8926         & 293.97      & \multirow{2}{*}{17}         & 0.26      & 0.00       & 1.000        & 14           & 0.37        \\
                      &                     &                      &                     & 1.001       & 3.77         & 0.53         & -            & -           &          & 0.69      & 0.62       & 1.001        & -            & -           \\ \cline{1-15}
\multirow{14}{*}{0.8} & \multirow{4}{*}{20} & 24                   & 30                  & 1.001       & 6.87         & 0.00         & 29           & 0.95        & 30         & 1.71      & 0.00       & 1.002        & 28           & 0.88        \\ \cline{3-15}
                      &                     & 32                   & 30                  & 1.002       & 7.15         & 0.00         & 5582         & 161.30      & 30         & 1.34      & 0.00       & 1.001        & 4527         & 140.48      \\ \cline{3-15}
                      &                     & \multirow{2}{*}{40}  & \multirow{2}{*}{22} & 1.001       & 5.96         & 0.00         & 24237        & 690.29      & \multirow{2}{*}{22}         & 1.39      & 0.00       & 1.001        & 8327         & 275.92      \\
                      &                     &                      &                     & 1.003       & 7.36         & 2.04         & -            & -           &           & 1.28      & 1.16       & 1.001        & -            & -           \\ \cline{2-15}
                      & \multirow{5}{*}{40} & 48                   & 30                  & 1.002       & 7.32         & 0.00         & 388          & 11.28       & 30         & 1.85      & 0.00       & 1.002        & 277          & 8.84        \\ \cline{3-15}
                      &                     & \multirow{2}{*}{64}  & \multirow{2}{*}{11} & 1.001       & 5.78         & 0.00         & 8191         & 239.93      & \multirow{2}{*}{11}         & 1.26      & 0.00       & 1.001        & 346          & 48.70       \\
                      &                     &                      &                     & 1.002       & 7.24         & 1.61         & -            & -           &          & 1.68      & 1.28       & 1.004        & -            & -           \\ \cline{3-15}
                      &                     & \multirow{2}{*}{80}  & \multirow{2}{*}{7}  & 1.000       & 5.95         & 0.00         & 21834        & 625.63      & \multirow{2}{*}{9}          & 0.68      & 0.00       & 1.000        & 182          & 4.90        \\
                      &                     &                      &                     & 1.001       & 7.31         & 1.21         & -            & -           &          & 1.19      & 1.09       & 1.001        & -            & -           \\ \cline{2-15}
                      & \multirow{5}{*}{60} & 72                   & 30                  & 1.003       & 7.51         & 0.00         & 6422         & 182.79      & 30         & 1.40      & 0.00       & 1.002        & 3105         & 149.92      \\ \cline{3-15}
                      &                     & \multirow{2}{*}{96}  & \multirow{2}{*}{3}  & 1.000       & 5.08         & 0.00         & 53358        & 1740.74     & \multirow{2}{*}{11}         & 0.11      & 0.00       & 1.000        & 23           & 0.49        \\
                      &                     &                      &                     & 1.001       & 7.43         & 1.18         & -            & -           &          & 1.30      & 1.22       & 1.001        & -            & -           \\ \cline{3-15}
                      &                     & \multirow{2}{*}{120} & \multirow{2}{*}{3}  & 1.002       & 5.19         & 0.00         & 45761        & 1563.13     & \multirow{2}{*}{13}         & 1.49      & 0.00       & 1.001        & 100          & 5.18        \\
                      &                     &                      &                     & 1.000       & 7.42         & 1.02         & -            & -           &          & 1.19      & 1.10       & 1.001        & -            & -           \\ \cline{1-15}
\multirow{14}{*}{0.7} & \multirow{4}{*}{20} & 24                   & 30                  & 1.001       & 9.66         & 0.00         & 26           & 0.91        & 30         & 2.75      & 0.00       & 1.003        & 22           & 0.84        \\ \cline{3-15}
                      &                     & 32                   & 30                  & 1.003       & 10.54        & 0.00         & 2380         & 65.99       & 30         & 1.95      & 0.00       & 1.002        & 958          & 58.25       \\ \cline{3-15}
                      &                     & \multirow{2}{*}{40}  & \multirow{2}{*}{17} & 1.001       & 10.27        & 0.00         & 15167        & 411.94      & \multirow{2}{*}{19}         & 1.29      & 0.00       & 1.001        & 8785         & 308.12      \\
                      &                     &                      &                     & 1.002       & 12.70        & 2.23         & -            & -           &          & 2.01      & 1.75       & 1.002        & -            & -           \\ \cline{2-15}
                      & \multirow{5}{*}{40} & 48                   & 30                  & 1.002       & 10.28        & 0.00         & 669          & 19.07       & 30         & 1.73      & 0.00       & 1.004        & 374          & 17.72       \\ \cline{3-15}
                      &                     & \multirow{2}{*}{64}  & \multirow{2}{*}{10} & 1.001       & 9.98         & 0.00         & 9469         & 276.89      & \multirow{2}{*}{12}         & 1.32      & 0.00       & 1.001        & 5151         & 174.52      \\
                      &                     &                      &                     & 1.001       & 10.79        & 2.15         & -            & -           &          & 2.47      & 2.24       & 1.002        & -            & -           \\ \cline{3-15}
                      &                     & \multirow{2}{*}{80}  & \multirow{2}{*}{4}  & 1.000       & 9.56         & 0.00         & 9839         & 252.58      & \multirow{2}{*}{11}         & 1.31      & 0.00       & 1.000        & 1757         & 80.37       \\
                      &                     &                      &                     & 1.001       & 10.38        & 1.13         & -            & -           &          & 1.39      & 1.29       & 1.001        & -            & -           \\ \cline{2-15}
                      & \multirow{5}{*}{60} & 72                   & 30                  & 1.004       & 9.58         & 0.00         & 7900         & 224.52      & 30         & 1.95      & 0.00       & 1.004        & 4776         & 197.43      \\ \cline{3-15}
                      &                     & \multirow{2}{*}{96}  & \multirow{2}{*}{3}  & 1.002       & 9.57         & 0.00         & 12456        & 313.75      & \multirow{2}{*}{10}         & 1.23      & 0.00       & 1.001        & 1887         & 76.50       \\
                      &                     &                      &                     & 1.001       & 10.94        & 1.52         & -            & -           &          & 1.82      & 1.64       & 1.002        & -            & -           \\ \cline{3-15}
                      &                     & \multirow{2}{*}{120} & \multirow{2}{*}{1}  & 1.000       & 10.45        & 0.00         & 14827        & 500.38      & \multirow{2}{*}{3}          & 1.68      & 0.00       & 1.000        & 9            & 2.09        \\
                      &                     &                      &                     & 1.001       & 11.62        & 0.91         & -            & -           &          & 1.39      & 1.29       & 1.001        & -            & -           \\ \cline{1-15}
\multicolumn{3}{|c||}{AVG}              & 499                 & 1.001       & 7.07         & 1.16         & 7132         & 201.50      & 570        & 1.28      & 0.36       & 1.002        & 2386         & 93.46     \\ \cline{1-15}
\end{tabular}
\end{center}
\end{sidewaystable}

\begin{sidewaystable}[h]
\renewcommand{\tabcolsep}{3pt}
\caption{Pattern B. $\beta = 0.7$} \label{TDRPP_alpha07}
\begin{center}
\scriptsize
\begin{tabular}{|c|c|c||c|c|c|c|c|c||c|c|c|c|c|c|}
\hline
\multicolumn{3}{|c||}{}                       & \multicolumn{6}{|c||}{ \cite{calogiuri2019branch} branch-and-bound  }& \multicolumn{6}{|c|}{ \cite{calogiuri2019branch} branch-and-bound with the new LB}                                                                                                          \\ \hline

\cline{1-15}
Delta                 & $|V|$               & $|A|$                & $OPT$          & $UB_I/LB_F$ & $GAP_I$ & $GAP_F$ & $NODES$ & $TIME$ & $OPT$ & $UB_I/LB_F$ & $GAP_I$ & $GAP_F$ & $NODES$ & $TIME$ \\ \cline{1-15}
\multirow{15}{*}{0.9} & \multirow{4}{*}{20} & 24                   & 30                  & 1.002       & 2.80         & 0.00         & 46           & 1.41        & 30         & 1.002       & 1.16       & 0.00       & 38           & 1.19        \\ \cline{3-15}
                      &                     & 32                   & 30                  & 1.002       & 3.10         & 0.00         & 11420        & 287.06      & 30         & 1.003       & 0.92       & 0.00       & 5156         & 236.80      \\ \cline{3-15}
                      &                     & \multirow{2}{*}{40}  & \multirow{2}{*}{5}  & 1.002       & 2.26         & 0.00         & 15285        & 426.92      & \multirow{2}{*}{9}          & 1.002       & 0.50       & 0.00       & 1767         & 153.81      \\
                      &                     &                      &                     & 1.002       & 2.94         & 1.00         & -            & -           &            & 1.002       & 0.87       & 0.64       & -            & -           \\ \cline{2-15}
                      & \multirow{5}{*}{40} & 48                   & 30                  & 1.002       & 2.63         & 0.00         & 889          & 22.72       & 30         & 1.001       & 0.97       & 0.00       & 622          & 20.67       \\ \cline{3-15}
                      &                     & \multirow{2}{*}{64}  & \multirow{2}{*}{2}  & 1.003       & 2.52         & 0.00         & 17516        & 481.34      & \multirow{2}{*}{6}          & 1.002       & 0.43       & 0.00       & 713          & 122.18      \\
                      &                     &                      &                     & 1.003       & 3.22         & 0.88         & -            & -           &            & 1.002       & 0.87       & 0.67       & -            & -           \\ \cline{3-15}
                      &                     & \multirow{2}{*}{80}  & \multirow{2}{*}{2}  & 1.001       & 2.57         & 0.00         & 10641        & 258.96      & \multirow{2}{*}{10}         & 1.085       & 0.36       & 0.00       & 2163         & 78.87       \\
                      &                     &                      &                     & 1.001       & 2.99         & 0.77         & -            & -           &            & 1.001       & 0.79       & 0.72       & -            & -           \\ \cline{2-15}
                      & \multirow{6}{*}{60} & \multirow{2}{*}{72}  & \multirow{2}{*}{27} & 1.003       & 2.88         & 0.00         & 21554        & 544.05      & \multirow{2}{*}{27}         & 1.003       & 0.93       & 0.00       & 10480        & 377.76      \\
                      &                     &                      &                     & 1.002       & 3.17         & 0.13         & -            & -           &            & 1.002       & 1.11       & 0.87       & -            & -           \\ \cline{3-15}
                      &                     & \multirow{2}{*}{96}  & \multirow{2}{*}{1}  & 1.001       & 2.63         & 0.00         & 43817        & 1029.50     & \multirow{2}{*}{9}          & 1.001       & 0.40       & 0.00       & 942          & 53.02       \\
                      &                     &                      &                     & 1.002       & 2.99         & 0.80         & -            & -           &            & 1.001       & 0.85       & 0.74       & -            & -           \\ \cline{3-15}
                      &                     & \multirow{2}{*}{120} & \multirow{2}{*}{1}  & 1.001       & 1.94         & 0.00         & 51691        & 1283.04     & \multirow{2}{*}{9}          & 1.000       & 0.34       & 0.00       & 56           & 158.40      \\
                      &                     &                      &                     & 1.001       & 3.05         & 0.81         & -            & -           &            & 1.001       & 0.87       & 0.81       & -            & -           \\ \cline{1-15}
\multirow{16}{*}{0.8} & \multirow{5}{*}{20} & 24                   & 30                  & 1.003       & 5.94         & 0.00         & 51           & 1.77        & 30         & 1.002       & 2.04       & 0.00       & 46           & 1.63        \\ \cline{3-15}
                      &                     & \multirow{2}{*}{32}  & \multirow{2}{*}{29} & 1.003       & 6.09         & 0.00         & 14295        & 417.48      & \multirow{2}{*}{29}         & 1.004       & 1.60       & 0.00       & 9416         & 391.95      \\
                      &                     &                      &                     & 1.002       & 6.35         & 1.92         & -            & -           &            & 1.001       & 1.80       & 1.72       & -            & -           \\ \cline{3-15}
                      &                     & \multirow{2}{*}{40}  & \multirow{2}{*}{4}  & 1.002       & 5.46         & 0.00         & 7475         & 207.79      & \multirow{2}{*}{9}          & 1.004       & 0.66       & 0.00       & 2763         & 110.56      \\
                      &                     &                      &                     & 1.005       & 6.25         & 2.07         & -            & -           &            & 1.003       & 1.74       & 1.40       & -            & -           \\ \cline{2-15}
                      & \multirow{5}{*}{40} & 48                   & 30                  & 1.003       & 6.10         & 0.00         & 1108         & 28.28       & 30         & 1.004       & 2.07       & 0.00       & 654          & 24.89       \\ \cline{3-15}
                      &                     & \multirow{2}{*}{64}  & \multirow{2}{*}{1}  & 1.002       & 3.96         & 0.00         & 9425         & 353.71      & \multirow{2}{*}{2}          & 1.803       & 1.41       & 0.00       & 5740         & 262.63      \\
                      &                     &                      &                     & 1.004       & 6.18         & 1.79         & -            & -           &            & 1.002       & 1.61       & 1.42       & -            & -           \\ \cline{3-15}
                      &                     & \multirow{2}{*}{80}  & \multirow{2}{*}{1}  & 1.002       & 4.71         & 0.00         & 676          & 24.31       & \multirow{2}{*}{4}          & 1.001       & 0.41       & 0.00       & 95           & 2.40        \\
                      &                     &                      &                     & 1.001       & 6.14         & 1.65         & -            & -           &            & 1.001       & 1.62       & 1.50       & -            & -           \\ \cline{2-15}
                      & \multirow{6}{*}{60} & \multirow{2}{*}{72}  & \multirow{2}{*}{28} & 1.005       & 6.27         & 0.00         & 17510        & 464.70      & \multirow{2}{*}{28}         & 1.006       & 1.93       & 0.00       & 10694        & 439.29      \\
                      &                     &                      &                     & 1.002       & 7.81         & 1.91         & -            & -           &            & 1.003       & 2.13       & 1.88       & -            & -           \\ \cline{3-15}
                      &                     & \multirow{2}{*}{96}  & \multirow{2}{*}{1}  & 1.001       & 3.82         & 0.00         & 44672        & 1354.95     & \multirow{2}{*}{2}          & 1.001       & 0.23       & 0.00       & 14041        & 450.08      \\
                      &                     &                      &                     & 1.001       & 6.22         & 1.61         & -            & -           &            & 1.003       & 1.57       & 1.30       & -            & -           \\ \cline{3-15}
                      &                     & \multirow{2}{*}{120} & \multirow{2}{*}{1}  & 1.001       & 4.72         & 0.00         & 94666        & 3480.03     & \multirow{2}{*}{5}          & 1.000       & 0.27       & 0.00       & 5            & 0.06        \\
                      &                     &                      &                     & 1.001       & 6.79         & 1.67         & -            & -           &            & 1.001       & 1.53       & 1.43       & -            & -           \\ \cline{1-15}
\multirow{16}{*}{0.7} & \multirow{5}{*}{20} & 24                   & 30                  & 1.004       & 8.83         & 0.00         & 43           & 1.27        & 30         & 1.004       & 3.01       & 0.00       & 38           & 1.19        \\ \cline{3-15}
                      &                     & \multirow{2}{*}{32}  & \multirow{2}{*}{29} & 1.005       & 9.62         & 0.00         & 8859.79      & 248.35      & \multirow{2}{*}{29}         & 1.005       & 2.19       & 0.00       & 4800         & 206.40      \\
                      &                     &                      &                     & 1.000       & 12.73        & 9.34         & -            & -           &            & 1.000       & 1.65       & 1.65       & -            & -           \\ \cline{3-15}
                      &                     & \multirow{2}{*}{40}  & \multirow{2}{*}{3}  & 1.001       & 8.94         & 0.00         & 29652        & 810.98      & \multirow{2}{*}{6}          & 1.003       & 0.42       & 0.00       & 7            & 0.09        \\
                      &                     &                      &                     & 1.005       & 9.55         & 2.98         & -            & -           &            & 1.003       & 2.08       & 1.80       & -            & -           \\ \cline{2-15}
                      & \multirow{5}{*}{40} & 48                   & 30                  & 1.008       & 9.35         & 0.00         & 1294         & 32.89       & 30         & 1.005       & 2.21       & 0.00       & 579          & 26.87       \\ \cline{3-15}
                      &                     & \multirow{2}{*}{64}  & \multirow{2}{*}{1}  & 1.002       & 6.17         & 0.00         & 15673        & 604.58      & \multirow{2}{*}{2}          & 1.000       & 0.10       & 0.00       & 3            & 0.25        \\
                      &                     &                      &                     & 1.002       & 9.17         & 3.00         & -            & -           &            & 1.003       & 2.31       & 1.96       & -            & -           \\ \cline{3-15}
                      &                     & \multirow{2}{*}{80}  & \multirow{2}{*}{1}  & 1.007       & 9.53         & 0.00         & 17939        & 662.74      & \multirow{2}{*}{2}          & 1.001       & 0.20       & 0.00       & 25           & 1.27        \\
                      &                     &                      &                     & 1.001       & 9.91         & 2.01         & -            & -           &            & 1.001       & 1.72       & 1.54       & -            & -           \\ \cline{2-15}
                      & \multirow{6}{*}{60} & \multirow{2}{*}{72}  & \multirow{2}{*}{26} & 1.008       & 8.72         & 0.00         & 18582        & 597.43      & \multirow{2}{*}{26}         & 1.010       & 2.63       & 0.00       & 12570        & 539.02      \\
                      &                     &                      &                     & 1.006       & 10.64        & 4.39         & -            & -           &            & 1.002       & 1.73       & 1.49       & -            & -           \\ \cline{3-15}
                      &                     & \multirow{2}{*}{96}  & \multirow{2}{*}{0}  & -           & -            & -            & -            & -           & \multirow{2}{*}{3}          & 1.002       & 0.28       & 0.00       & 15593        & 670.18      \\
                      &                     &                      &                     & 1.002       & 9.41         & 2.49         & -            & -           &            & 1.002       & 1.94       & 1.78       & -            & -           \\ \cline{3-15}
                      &                     & \multirow{2}{*}{120} & \multirow{2}{*}{0}  & -           & -            & -            & -            & -           & \multirow{2}{*}{2}          & 1.001       & 0.50       & 0.00       & 31           & 0.56        \\
                      &                     &                      &                     & 1.001       & 9.53         & 1.74         & -            & -           &            & 1.001       & 1.79       & 1.63       & -            & -           \\ \cline{1-15}
\multicolumn{3}{|c||}{AVG}                               & 373                 & 1.003       & 6.16         & 1.72         & 8586         & 239.89      & 429        & 1.006       & 1.54       & 0.63       & 3958         & 169.04      \\ \cline{1-15}
\end{tabular}
\end{center}
\end{sidewaystable}


  \section{Acknowledgments}
\noindent This research was partly supported by the Ministero delI'Istruzione, dell'Universit\`{a} e della Ricerca (MIUR) of Italy (PRIN project  2015JJLC3E\_005 ``Transportation and Logistics Optimization in the Era of Big and Open Data''). This support is gratefully acknowledged.

\bibliography{biblio.bib}
\bibliographystyle{plainnat}
\end{document}